%% file: main-corrected.tex
\documentclass[reprint,
superscriptaddress,
amsmath,amssymb,
aps,
]{revtex4-2}

\input{base}

\tikzset{gg/.style={opacity=0, fill opacity=1,rounded corners, fill=blue!10, inner xsep=1pt, inner ysep=6pt, thin}, 
	pg/.style={opacity=0, fill opacity=1,rounded corners, fill=green!10, inner xsep=1pt, inner ysep=0pt, thin},
	dg/.style={opacity=0, fill opacity=1,rounded corners, fill=red!10, inner xsep=1pt, inner ysep=0pt, thin}}
\newcommand{\metershift}{0.8mm}

\newcommand{\refappendix}[1]{\hyperref[#1]{Appendix~\ref*{#1}}}

\begin{document}
	
	\preprint{APS/123-QED}
	
	\title{Photon catalysis for general multimode multi-photon quantum state preparation}
	
	\author{Andrei Aralov}
	\thanks{These two authors contributed equally.}
	\affiliation{Laboratoire Kastler Brossel, Sorbonne Universit\'{e}, CNRS, ENS-Universit\'{e} PSL,  Coll\`{e}ge de France, 4 place Jussieu, F-75252 Paris, France}
	\author{Émilie Gillet}
	\thanks{These two authors contributed equally.}
	\affiliation{Laboratoire Kastler Brossel, Sorbonne Universit\'{e}, CNRS, ENS-Universit\'{e} PSL,  Coll\`{e}ge de France, 4 place Jussieu, F-75252 Paris, France}
	\author{Viet Nguyen}
	\affiliation{Sorbonne Université, F-75006 Paris, France}
	\author{Andrea Cosentino}
	\affiliation{Sorbonne Université, F-75006 Paris, France}
	\author{Mattia Walschaers}
	\email{mattia.walschaers@lkb.upmc.fr}
	\affiliation{Laboratoire Kastler Brossel, Sorbonne Universit\'{e}, CNRS, ENS-Universit\'{e} PSL,  Coll\`{e}ge de France, 4 place Jussieu, F-75252 Paris, France}
	\author{Massimo Frigerio}
	\email{massimo.frigerio@lkb.upmc.fr}
	\affiliation{Laboratoire Kastler Brossel, Sorbonne Universit\'{e}, CNRS, ENS-Universit\'{e} PSL,  Coll\`{e}ge de France, 4 place Jussieu, F-75252 Paris, France}

	%
	%
	
	\date{\today}
	
	\begin{abstract}
		Multimode multiphoton states are at the center of many photonic quantum technologies, from photonic quantum computing to quantum sensing. In this work, we derive a procedure to generate exactly, and with a predictable number of steps, any such state by using only multiport interferometers, photon number resolving detectors, photon additions and displacements. We achieve this goal by establishing a connection between photonic quantum state engineering and the algebraic problem of symmetric tensor decomposition. This connection allows us to solve the problem by using corresponding results from algebraic geometry and unveils a mechanism of photon catalysis, where photons are injected and subsequently retrieved in measurements, to generate entanglement that cannot be obtained through Gaussian operations. We also introduce a tensor decomposition, that generalizes our method and allows to construct circuits yielding perfect fidelity, using the minimum number of catalysis photons. As a benchmark, we numerically evaluate our method and compare its performance with state-of-the art results, confirming 100\% fidelity on different classes of states.
	\end{abstract}
	
	\maketitle
	

	\section{Introduction}
	Multiphoton states of a multimode system are versatile and widely studied finite linear subspaces of the Fock space of multimode bosonic systems \cite{Pan2012} and they find applications in quantum computation \cite{KLM2001,LOQCrev,Flamini_2019,PhysRevLett.130.090602,Fusion2023,Quandela24,PhysRevResearch.7.033051,alexander_manufacturable_2025,Salavrakos_2025}, quantum metrology \cite{Dowling2008, PhysRevLett.102.040403,Xiang2013,PhysRevLett.111.070403,PhysRevA.95.032321,NOONLim,Namkung_2024} quantum simulation \cite{PhysRevLett.120.130501,10.1063/5.0181151,Jelmer} as well as being at the core of various types of boson sampling tasks \cite{10.1145/1993636.1993682, PhysRevA.96.032326,Spagnolo2023}. Moreover, these same states have been introduced under the name of \emph{core states} \cite{Chabaud2022holomorphic} in the stellar representation of non-Gaussian states and their characterization is therefore essential to the understanding of non-Gaussian states with finite stellar rank \cite{StellarRepr}. 
	
	When the total number of photons is constrained, the corresponding subspaces are invariant under all linear optical operations (equivalently, all passive Gaussian operations). However, the vast majority of them cannot be prepared starting from factorized multimode Fock states and passive linear optics, nor, more generally, by alternating probabilistic photon additions and interferometers, as can be easily deduced by counting free parameters. This aspect of multiphoton quantum states has been implicitly known for a long time; standard references on bosonic systems often emphasize the need to consider superpositions of such states to construct the full Fock space \cite{Bratteli1997-bc,Verbeure2010-gt,PRXQuantum.2.030204}. The issue is also explicitly discussed in \cite{PhysRevResearch.3.033018}. A first explicit scheme to construct arbitrary multimode multiphoton states dates back several decades \cite{PhysRevA.68.042325}. This scheme relies on intense coherent states which are truncated using quantum scissors \cite{PhysRevLett.81.1604}, and intricate Bell-type measurements \cite{DUSEK2001161} that must be implemented through postselection. While this approach works in principle, it is highly resource intensive and does not come with a mathematical framework that allows to reason about its optimality. A more structured classification of multimode states with a definite number of photons was recently proposed in \cite{Kopylov:2025zcs}, where the authors show that only those belonging to a specific class can be generated by alternating $N$ photon additions with passive linear optics transformations acting on the $M$ target modes alone, but how to generate states belonging to the other classes remained an open problem, despite the fact that they constitute the vast majority of multimode multiphoton states.
	
	In this paper, we tackle the very general problem of exactly preparing any multimode multiphoton state by considering an ancillary mode on which we perform post-selection through photon counting at the end of the process. In this way, we show that the $M$ modes can be projected into \emph{any} desired state with $N$ photons, with a quantifiable cost and a fidelity of one. By adding a second ancillary mode, moreover, we can also prepare any state with non-homogeneous photon numbers, i.e., any multimode multiphoton state.
	
	Our protocol can effectively be implemented with a specific realization of a boson sampling \cite{10.1145/1993636.1993682} or Gaussian boson sampling \cite{PhysRevLett.119.170501} device. In order to produce a state with up to $N$ photons, the boson sampler that implements our state-engineering protocol must generally be seeded with much more than $N$ photons. Even though the additional photons are retrieved by the photon counter, they serve the key purpose of creating a type of mode-intrinsic entanglement \cite{PhysRevA.100.062129,Lopetegui:25}. This process has been known in literature under the name of photon catalysis \cite{PhysRevLett.88.250401}. Hitherto, it was generally applied to produce a range of single-mode states \cite{PhysRevA.86.043820, Hu_2016, Birrittella:18, Eaton_2019}, but its role in the increase of entanglement has also been investigated \cite{Bartley_2015}. Recently, photon catalysis was studied in a multimode context, but only for two-photon states \cite{10821417}. Our present work can be seen as a proof by construction that photon catalysis can create any multimode multi-photon state. The two-photon result of \cite{10821417} is retrieved as a special case of the proposed protocol.
	
	In the alternative implementation with a Gaussian boson sampler, we inject a series of weakly squeezed states into an interferometer and subsequently count photons. This procedure is commonly used to create single-mode non-Gaussian states for continuous variable quantum computing \cite{PhysRevA.100.052301,Bourassa2021blueprintscalable,PhysRevLett.128.240503,larsen_integrated_2025}. While possible application for engineering multimode states are discussed in \cite{PhysRevA.100.052301}, no general protocol was provided. In the context of continuous-variable measurement-based quantum computing \cite{PhysRevLett.97.110501}, it is implicitly known that we can obtain any multimode state by performing measurements on a well-chosen Gaussian state \cite{arzani2025effectivedescriptionsbosonicsystems}. Our work can be seen as the first explicit demonstration of such a general quantum state-engineering protocol. 
	
	The layout of our paper is as follows: after an introduction to the theory of multimode multiphoton states, also known as core states, and their representation as stellar polynomials in \autoref{sec:stellar-repr}, we introduce our main result in \autoref{sec:main-result} and prove it using Waring decomposition and monic expansion. In \autoref{sec:photonic-impl}, we explain how our protocol could be effectively implemented with common photonic setups: a weak version of a boson sampler and a Gaussian boson sampler. In \autoref{sec:bounds}, we provide results bounding the number of photon additions and the rank of the Fock state on which the ancillary mode has to be projected and, in \autoref{sec:special-cases}, we analyze some special cases: an example of GHZ state preparation inspired by the proposed general methods but trading fewer additions for extra displacements; and an analytical solution of the elementary symmetric polynomial (ESP) decomposition for $d=2$. Finally, in \autoref{sec:numerical-methods}, a software implementation of our method and its variants allows us to compare fidelity, probability of success and resource use for a selection of state preparation tasks.
	
	\section{Stellar Representation And Core States}
	\label{sec:stellar-repr}
	
	In this section we introduce the stellar representation of quantum states that will allow us to map our set of quantum operations on core states as operations on complex multivariate polynomials \cite{StellarRepr, Chabaud2022holomorphic}. The stellar function of a given quantum state of $M$ modes $\ket{\psi} \in \mathcal{H}^{\otimes M}$, where $\mathcal{H} \simeq \mathrm{L}^2(\mathbb{R})$ is the Fock space of one bosonic mode, is defined as a renormalized overlap of a multimode coherent state and $\ket{\psi}$:
	\begin{equation}
		F^\star_\psi ( \boldsymbol{\alpha} ) = e^{\frac{1}{2}\modsq{\boldsymbol{\alpha}}} \dbraket{\boldsymbol{\alpha}^*}{\psi}
	\end{equation}
	where $\boldsymbol{\alpha} \in \mathbb{C}^M$. 
	It can be shown \cite{StellarRepr} that any state can be written as:
	\begin{equation}
		\ket{\psi} = F^\star_\psi\left( \mbf{\adhat} \right)\ket{\mbf{0}}.
	\end{equation}
	where $\mbf{\adhat} = (\adhat_{1}, \dots ,\adhat_{M})$ are the creation operators of the $M$ modes, satisfying the canonical commutation relations $[ \ahat_{j}, \adhat_{k} ] = \delta_{jk} $ with their adjoints. 
	
	In particular, for core states, which have a finite support in the Fock basis, the stellar function becomes a polynomial with complex coefficients $P_\psi \in \mathbb{C}\left[ x_1, \dots , x_M \right]$, where from now on we denote by $\mathbb{C}\left[ x_1, \dots , x_M \right]$ the ring of polynomials in $M$ variables (of arbitrary degree) on the complex field.
	\begin{align}
		\ket{\psi} = \sum_{\mbf{n}} \psi_{\mbf{n}} \ket{\mbf{n}} &= \sum_{\mbf{n}} \frac{\psi_{\mbf{n}}}{\sqrt{\mbf{n}!}} (\mbf{\adhat})^{\mbf{n}} \ket{\mbf{0}} = P_\psi\left(\mbf{\adhat}\right)\ket{\mbf{0}} \\
		P_\psi\left( \mbf{x} \right) &= \sum p_{\mbf{n}} \mbf{x}^{\mbf{n}}; \;\;\; p_\mbf{n} = \frac{\psi_\mbf{n}}{\sqrt{\mbf{n}!}}
	\end{align}
	
	Here and throughout the paper we utilize the usual multi-index notation:
	\begin{alignat}{2}
		& \mbf{n} = \left(n_1 ,\dots , n_M\right) \quad && \mbf{x} = \left[x_1 , \dots , x_M\right]^T \\
		& \mbf{n}! = n_1! \dots n_M! \quad && \left|\mbf{n}\right| = \sum_{i = 1}^M n_i \\
		& \ket{\mbf{n}} = \ket{n_1 \dots n_M} \quad && \mbf{x}^\mbf{n} = x_1^{n_1} \dots x_M^{n_M}
	\end{alignat}
	
	This bijective correspondence allows us to study multivariate polynomials $P_\psi$ and then transfer results to quantum states (\autoref{tbl:state-poly-corresp}). For simplicity of notation, we will omit normalization factors as they are easily applied by scaling all the coefficients of the polynomial.
	
	\begin{table}[ht]
		\renewcommand{\arraystretch}{1.5}
		\centering
		\begin{tabular}{c|c}
			State transformation & Polynomial transformation \\ \hline
			$\adhat_k\ket{\psi}$ & $x_k P_\psi\left(\mbf{x}\right)$ \\
			$\ahat_k\ket{\psi}$ & $\frac{\partial}{\partial x_k} P_\psi\left(\mbf{x}\right)$ \\
			$\bra{n}_k\ket{\psi}$ & $\frac{\partial^n}{\partial x_k^n} P_\psi\left(\mbf{x}\right)\Bigr|_{x_k = 0}$
		\end{tabular}
		\caption{Correspondence of state transformations and operations on polynomials}
		\label{tbl:state-poly-corresp}
	\end{table}
	
	Another noteworthy reason to study stellar polynomials is that one can classify core states by the number of irreducible factors of different degrees in their factorization. In particular, it is possible to prepare any core state whose stellar polynomial is homogeneous and factorizes into a product of $N$ linear terms directly acting on the $M$ target modes, starting in the vacuum state, by alternating photon additions and multiport interferometers without ancillary modes, as shown in \cite{Kopylov:2025zcs}. The class of such states is denoted $\left[1^N\right]_M$.
	
	The question remains, however, how and whether one can prepare all other classes of irreducible polynomials with these resources, by adding one ancillary mode to the $M$ target modes and conditioning on a specific photon counting outcome of that mode.

	\section{Preparation Of Arbitrary Multimode Core States}
	\label{sec:main-result}
	In this section we present our main contribution, a method to prepare any quantum state with finite support in the Fock basis using one or two ancillary modes, multiport interferometers, photon additions, photon-number-resolving (PNR) detectors and displacement operations. We start from the aforementioned observation that any $M$-modes multiphoton state described by a homogeneous polynomial in $\left[1^N\right]_M$ can be prepared exactly with photon additions and multiport interferometers acting solely on the $M$ modes. However, in the general case and unlike single-variable and two-variables polynomials, not all polynomials in $M \geq 3$ variables factorize into  a product of linear terms, which leads to a natural question: what is the amount of resources needed to prepare an arbitrary core state? 
	The method that we propose guarantees the preparation of an arbitrary core state of $d$ photons in $M$ modes with 100\% fidelity in a finite number of steps, and quantifies the required resources in terms of overhead of photon additions to be performed. The proposed method is described in the following \autoref{th:main-theorem} and depicted in \autoref{fig:main-circuit}.
	\begin{theorem}
		\label{th:main-theorem}
		Any core state containing at most $d$ photons in $M$ modes, described by a polynomial $P_\psi \in \mathbb{C}\left[ x_1 , \dots , x_M \right]$ of degree $d$ in $M$ variables, whose homogeneization has Waring rank $r$ can be prepared exactly from the multimode vacuum state, using two ancillary modes in the vacuum, $d\times r$ multiport interferometers and photon additions, one photon projection conditioning on having $d\left(r - 1\right)$ photons, and additionally, in case of an inhomogeneous state, one displacement and one projection conditioning on having zero photons.
	\end{theorem}
	The constructive proof of this theorem is detailed in the following subsections and it is based on three main steps. First, we show how any homogeneous polynomial in $M$ variables can be seen as the coefficient of some power of an ancillary variable $\lambda$ inside a polynomial in $M+1$ variables. This coefficient, hence the target core state, can be extracted from the corresponding superposition state of $M+1$ modes by post-selecting on measuring the right number of photons in the ancillary mode associated with $\lambda$.  Secondly, we show that this enlarged polynomial factors as a product of linear terms, therefore the corresponding core state can be prepared with photon additions and interferometers acting solely on the $M+1$ modes. Finally, we show that adding a second ancillary mode, displacing it by a fixed amount and projecting on the vacuum leads to the preparation of core states with variable numbers of photons in the superposition.
	
	While \autoref{th:main-theorem} provides a general procedure to prepare any core state and relates it with a well studied tensor decomposition, it could be demonstrated that it is not optimal in terms of the number of catalysis photons and the order of the PNR projection required (see \autoref{sec:numerical-methods}). Below we introduce a new tensor decomposition addressing the optimality of the preparation scheme, which is surprisingly equivalent to the universality of the classical algebraic circuit model introduced in \cite{SHPILKA2002639}.
	\begin{theorem}[Elementary Symmetric Decomposition]
		\label{th:elementary-symmetric-decomposition}
		There exists $N$ such that for any homogeneous polynomial $P \in \mathbb{C}\left[x_1, \dots, x_M\right]$ of degree $d$ there is a matrix $U: \mathbb{C}^M \rightarrow \mathbb{C}^N$ such that  the polynomial could be written as an ESP of the set of linear forms defined by this matrix:
		\begin{equation}
			P\left(\mbf{x}\right) = e_d\left(U\mbf{x}\right)
		\end{equation}
		
		In particular, for a core state $\ket{\psi}$, when the decomposition of the homogenized stellar polynomial is found $P_\psi^h\left(\mbf{x}\right) = e_d\left(U\mbf{x}\right)$, the state could be prepared using $N$ multiport interferometers and photon additions, and one photon projection conditioning on having $N - d$ photons, and additionally, in case of an inhomogeneous state, one displacement and one projection conditioning on having zero photons.
		
		Furthermore, under certain conditions this decomposition leads to the shortest state preparation protocol involving the given resources (photon additions, multiport interferometers and PNR detection on a single ancillary mode). These conditions are the following:
		\begin{enumerate}
			\item The stellar polynomial is homogeneous.
			\item The stellar polynomial is irreducible.
		\end{enumerate}
	\end{theorem}
	
	As described in \autoref{sec:numerical-methods}, the decomposition introduced in \autoref{th:elementary-symmetric-decomposition} is well-suited for numerical optimization (such as gradient descent) which makes it the preferable tool for practical implementations. We also note that symbolic methods based on Gröbner bases can, in principle, be used to determine the minimal $N$ for which a solution to this decomposition exists. However, the computational cost of constructing these bases grows rapidly with the number of variables and the degree of the polynomials, which in this context corresponds to the number of modes, and of photons.
	
	\begin{figure*}[t]
		\centering
		\begin{quantikz}[column sep=0.5cm]
			\lstick[5]{$\ket{0}^{\otimes (M+2)}$}
			& \gate{\adhat}\gategroup[5, steps=5, background, style=gg, label style={yshift=1mm}]{$\left[1^{N}\right]_{M+2}$} & \gate[5]{\hat{U}_1} & {\ \dots \ } & \gate{\adhat} & \gate[5]{\hat{U}_{N}}
			& [-0.15cm] \slice[style={xshift=1cm}, label style={pos=1, anchor=north}, style={very thick}]{
				$\ket{L_{\psi^h}}$
			}
			&& [-0.15cm] \meter[label style={yshift=\metershift}]{\ketbra{n}{n}}\gategroup[1, steps=1, background, style=pg]{} &
			\setwiretype{n} \slice[label style={pos=0, anchor=south west, xshift=-4mm}, style={very thick}]{
				$\ket{{\psi^h}}$ 
			}
			& \\
			\
			& 				&				& {\ \dots \ } &&&&&\gategroup[1, steps=4, background, style=dg]{}&& \gate{\hat{D} \left(1\right)} & \meter[label style={yshift=\metershift}]{\text{$\ketbra{0}{0}$}} \\
			& 				&				& {\ \dots \ } &&&&&&&&\rstick[3]{$\ket{\psi}$} \\
			\
			& \setwiretype{n} \vdots &		& {\ \dots \ } &&&&& \vdots &&&\\
			\
			& 				&				& {\ \dots \ } &&&&&&&&
		\end{quantikz}
		\caption{Circuit of the proposed method. By injecting $N$ photons, $n$ of which will be later projected-out, a state of stellar rank $d = N - n$ is obtained. The general method described in \autoref{th:main-theorem} sets $N = dr, n = d\left(r - 1\right)$; while the special case described in \autoref{th:e2-corollary} sets $N = M, n = M - 2$. The blue part represents the preparation of the seed state, the stellar polynomial of which factorizes into a product of linear terms, the green part represents the post-selection that allows to construct an arbitrary homogeneous state and the red part represents dehomogenization.}
		\label{fig:main-circuit}

	\end{figure*}
	
	\subsection{Waring decomposition}
	
	As a first step, we recall the correspondence between homogeneous polynomials and symmetric tensors \cite{Landsberg2012_PolyTens}: for any homogeneous polynomial $P_\psi$ of degree $d$, there is a corresponding symmetric tensor $\mathcal{P}_\psi \in S^d\left(\mathbb{C}^M\right)$ encoding its coefficients in a monomial basis. In particular, this correspondence allows to relate the decomposition of a polynomial into a sum of powers of linear forms to symmetric tensor decomposition. Every symmetric tensor could be decomposed as a sum of tensorial powers of some vectors $\mbf{w}_k \in \mathbb{C}^M$, where the number of terms in the decomposition is called the symmetric rank of $\mathcal{P_\psi}$.
	\begin{equation}
		\mathcal{P_\psi} = \sum_{k=1}^r \mbf{w}^{\otimes d}_k
	\end{equation}

	The corresponding homogeneous polynomial decomposition is called the \emph{Waring decomposition}, and the number of terms $r$ is the \emph{Waring rank} of that polynomial.
	\begin{equation}
		P_\psi\left(\mbf{x}\right) = \sum_{k=1}^r \left( \mbf{w}_k^T \mbf{x} \right)^d
	\end{equation}
	
	\subsection{Monic expansion and dehomogenization}
	We now notice that the expression in the Waring decomposition of $P_\psi$ can be trivially rewritten as $e_{1}\left( \left(\mbf{w}_1^T\mbf{x}\right)^d, \dots, \left(\mbf{w}_r^T\mbf{x}\right)^d \right)$, where: 
	\begin{equation}
		e_1\left(\mbf{x}\right) = \sum_{j=1}^M x_j
	\end{equation}
	is the first elementary symmetric polynomial. Pursuing this idea, we can prove the following result:
	\begin{lemma}
		\label{th:monic-and-waring-0}
		Let $P_{\psi}(\mbf{x}) \in \mathbb{C}\left[x_1, \dots, x_M\right]$ be a homogeneous polynomial of degree $d$, having Waring decomposition $P_\psi\left(\mbf{x}\right) = \sum_{k=1}^r \left( \mbf{w}_k^T \mbf{x} \right)^d$. Define
		\begin{align}
			W &= \begin{bmatrix} \mbf{w}_1  ,  \dots , \mbf{w}_r \end{bmatrix}^T \\
			L_{r, d, W}\left(\lambda, \mbf{x}\right) &= \prod_{k=1}^r \left( \lambda^d + \left( \mbf{w}_k^T\mbf{x} \right)^d \right)
		\end{align}
		Then the coefficient of $\lambda^{d(r-1)}$ in the expansion of $L_{r, d, W}\left(\lambda,\mbf{x}\right)$ is exactly $P_{\psi}(\mbf{x})$. 
	\end{lemma}
	
	\begin{proof}
		We start by recalling the following formula for the so-called \emph{monic expansion}, a direct corollary of Vieta's formula \cite{Cox2015ESP}:
		\begin{equation}\label{eq:monic-expansion}
			\prod_{k=1}^M \left( y + x_k \right) = \sum_{k=0}^M y^{M-k} e_k\left(x_1, \dots, x_M\right)
		\end{equation}
		\begin{equation}
			e_{k}\left(\mbf{x}\right) = \sum_{1 \leq j_1 < \dots < j_k \leq M} x_{j_1} \dots x_{j_k}
		\end{equation}
		where $e_k$ is a $k$-th elementary symmetric polynomial. We can therefore relate the monic expansion to the Waring decomposition by applying the first to: 
		\begin{multline}
			\prod_{k=1}^r \left( \lambda^d + \left( \mbf{w}_k^T\mbf{x} \right)^d \right) \ = \\
			= \ \sum_{k=0}^{r} \lambda^{d\left(r - k\right)}e_k\left( \left(\mbf{w}_1^T\mbf{x}\right)^d, \dots, \left(\mbf{w}_r^T\mbf{x}\right)^d \right) \label{eq:L-sum}
		\end{multline}
		and noticing that the second term ($k=1$) in the expansion in \autoref{eq:L-sum} contains exactly the Waring decomposition of rank $r$.
		\begin{multline}
			\label{eq:l-first-term}
			\lambda^{d\left(r - 1\right)}e_1\left( \left(\mbf{w}_1^T\mbf{x}\right)^d, \dots, \left(\mbf{w}_r^T\mbf{x}\right)^d \right) =\\ 
			= \ \lambda^{d\left(r - 1\right)}\left( \sum_{i = 1}^r \left( \mbf{w}_i^T \mbf{x} \right)^d \right)
		\end{multline}
	\end{proof}
	To obtain our result for homogeneous polynomials, the last missing step is to show that each factor $\left( \lambda^d + \left( \mbf{w}_k^T\mbf{x} \right)^d \right)$ in \autoref{eq:L-sum} can be decomposed into a product of linear factors:
	\begin{lemma}
		\label{th:monic-and-waring}
		Let $\omega_n = e^{\frac{2 \pi i}{n}}$ be the primitive n-th root of unity. Then we have:
		\begin{equation}
			L_{r, d, W} = \prod_{k=1}^r \prod_{i=0}^{d-1} \left( \lambda - \omega_{2d}\omega_d^i \mbf{w}_k^T\mbf{x} \right) \in \left[ 1^{dr} \right]_{M} \label{eq:L-prod}
		\end{equation}
	\end{lemma}

	We also notice that the above construction can easily be generalized to the preparation of arbitrary, non-homogeneous polynomials from homogeneous ones by a so-called \emph{dehomogenization} \cite{Cox2015}.
	\begin{lemma}
		\label{th:dehomogeniztion}
		For any $P \in \mathbb{C}\left[x_1, \dots, x_M\right]$ there exists $P^h \in \mathbb{C}\left[y, x_1, \dots, x_M\right]$; $P^h$ -- homogeneous, such that $\forall \mbf{x} \in \mathbb{C}^M: P\left(\mbf{x}\right) = P^h\left(1, \mbf{x}\right)$.
		Physically, one can prepare any core state involving superpositions of variable numbers of photons by first preparing the homogeneous state $\vert \psi^h \rangle$ corresponding to $P^h$, then displacing the ancillary mode used for dehomogenization by $1$ and finally conditioning on measuring $0$ photons in this ancillary mode: $\ket{\psi} = \bra{0}_0 \hat{D}_0\left( 1 \right) \ket{\psi^h}$.
	\end{lemma}
	
	\subsection{Proof of \autoref{th:main-theorem}}
	\label{sec:main-proof}
	Finally, we can piece together these results to prove our main theorem:
	\begin{proof}
		Let $P_\psi\left(\mbf{x}\right) = \sum p_\mbf{n} \mbf{x}^\mbf{n}$ of degree $d$, not necessarily homogeneous, be the stellar polynomial of the target core state $\vert \psi \rangle$. Denote by $P_{\psi^h}\left( \mbf{y}\right)$ its homogenization ($\mbf{y} =\left[y, x_{1},\dots,x_M\right]^T \in \mathbb{C}^{M+1}$) and compute its Waring decomposition, yielding the Waring rank $r$ and the vectors $\mbf{w}_k \in \mathbb{C}^{M + 1}$ such that
		\begin{equation}
			P_{\psi^h}\left(\mbf{y}\right) = \sum_{k=1}^r \left( \mbf{w}_k^T\mbf{y} \right)^d
		\end{equation}
		
		Since, by \autoref{th:monic-and-waring}, with the addition of an ancillary variable $\lambda$ we have that $L_{\psi^h} = L_{r, d, W}\left(\lambda, \mbf{y}\right) \in \left[1^N\right]_M$, we can prepare the corresponding core state using the method described in \cite{Kopylov:2025zcs} (included here in \refappendix{App:VectorToUnitary}), starting from the multimode vacuum state and alternating photon additions with multiport interferometers. Then, according to \autoref{tbl:state-poly-corresp}, conditioning the ancillary mode $\ahat_{\lambda}$ on having $d\left(r-1\right)$ photons corresponds to collapsing the sum onto the coefficient next to $\lambda^{d\left(r-1\right)}$, which by \autoref{eq:l-first-term} gives us $P_{\psi^h}$.
		
		Finally, by \autoref{th:dehomogeniztion}, applying displacement and vacuum-projection on mode $\adhat_y$ to $P_{\psi^h}\left(\mbf{y}\right)$ produces the dehomogenized polynomial $P\left(\mbf{x}\right)$.
	\end{proof}
	
	\section{Photonic implementation of the general scheme}
	\label{sec:photonic-impl}
	
	\begin{figure*}[t]
		\centering
		\subfloat[\label{fig:main-circuit-boson}]{
			\tikzset{ggext/.style={opacity=0, fill opacity=1, rounded corners, fill=blue!10, inner xsep=2pt, inner ysep=6pt, thin, xshift=-1cm, yshift=2mm}}
			\centering
			\begin{quantikz}[column sep=0.3cm]
				\lstick[3]{$\ket{1}^{\otimes N}$} \gategroup[3, steps=6, background, style=ggext]{} & & \gate[7][1cm]{\hat{\bm{U}}} && \meter{\text{$\ketbra{0}{0}$}} \\
				& \setwiretype{n} \vdots & & \setwiretype{q} & \meter{\text{$\ketbra{0}{0}$}} \\ 
				& & \gategroup[5, steps=1, background, style=gg]{} && \meter{\text{$\ketbra{0}{0}$}} \\[0.5cm]
				\lstick[4]{$\ket{0}^{\otimes \left(M + 2\right)}$} & & & & \meter{\ketbra{n}{n}}\gategroup[1, steps=1, background, style=pg]{} \\
				& & & & \gate{\hat{D} \left(1\right)}\gategroup[1, steps=2, background, style=dg]{} & \meter[label style={yshift=\metershift}]{\text{$\ketbra{0}{0}$}} \\
				& \setwiretype{n} \vdots & & \setwiretype{q} & &\rstick[2]{$\ket{\psi}$} \\
				& & & & &
			\end{quantikz}
		}
		\subfloat[\label{fig:main-circuit-gaussian-boson}]{
			\tikzset{ggext/.style={opacity=0, fill opacity=1, rounded corners, fill=blue!10, inner xsep=2pt, inner ysep=6pt, thin, xshift=-1cm, yshift=2mm}}
			\centering
			\begin{quantikz}[column sep=0.3cm]
				\lstick[3]{$\ket{0}^{\otimes N}$} \gategroup[3, steps=6, background, style=ggext]{} & & \gate[7][1cm]{\hat{\bm{G}}} && \meter{\text{$\ketbra{1}{1}$}} \\
				& \setwiretype{n} \vdots & &  \setwiretype{q} & \meter{\text{$\ketbra{1}{1}$}} \\ 
				& & \gategroup[5, steps=1, background, style=gg]{} && \meter{\text{$\ketbra{1}{1}$}} \\[0.5cm]
				\lstick[4]{$\ket{0}^{\otimes \left(M + 2\right)}$} & & & & \meter{\ketbra{n}{n}}\gategroup[1, steps=1, background, style=pg]{} \\
				& & & & \gate{\hat{D} \left(1\right)}\gategroup[1, steps=2, background, style=dg]{} & \meter[label style={yshift=\metershift}]{\text{$\ketbra{0}{0}$}} \\
				& \setwiretype{n} \vdots & & \setwiretype{q} & &\rstick[2]{$\ket{\psi}$} \\
				& & & & &
			\end{quantikz}
		}
		\caption{Alternative representations of the proposed method, where approximate photon additions are implemented through injection of single photons coupled by weakly reflecting beam-splitters (\autoref{fig:main-circuit-boson}) as in a boson sampler, or by weak single-mode squeezing operations followed by heralding of single photons (\autoref{fig:main-circuit-gaussian-boson}) as in a Gaussian boson sampler. }\label{fig:main-circuit-alt}
	\end{figure*}
	
	After having outlined the way in which the preparation of a generic multimode core state can be decomposed into an alternating sequence of photon additions and interferometers, we shall see how this construction relates to other photonic quantum information protocols by considering in greater detail the physical implementation of the photon addition steps. As a first remark, notice that the multiport interferometers ($U_1$ to $U_N$ in \autoref{fig:main-circuit}) effectively allow us to perform photon addition on arbitrary linear combination of modes, and this is their only purpose in our construction. Moreover, a possible way to implement a photon addition on a state $\vert \psi \rangle$ is to let this state interact with the first Fock state $\vert 1 \rangle$ of another mode through a beam-splitter, whose unitary evolution is described by $\hat{U}_{\mathrm{BS}}(\theta) = e^{  \theta (\hat{a}^{\dagger}_{1} \hat{a}_{2} - \hat{a}_{1} \hat{a}^{\dagger}_{2} )}$, where $\hat{a}_{1},\hat{a}_{2}$ are the annihilation operators of the modes described by the state $\vert \psi \rangle$ and $\vert 1 \rangle$, respectively. If $\theta$ is sufficiently small and if the output two-mode state is conditioned upon measuring the vacuum on the second mode, the resulting state will be (see \refappendix{sec:asymptotic-photon-addition}):
	\begin{equation}
		\langle 0 \vert_{B} \hat{U}_{\mathrm{BS}}(\theta)  \vert \psi \rangle \otimes \vert 1 \rangle  \ = \ \theta \left[ \hat{a}^{\dagger}_{1} \vert \psi \rangle  \ + \ \mathcal{O}(\theta^2) \right]
	\end{equation}
	which is an approximation to a photon-added state. The probability with which this protocol will work is given by the squared norm of this output state, $p_0 = \theta^2 (\langle \adhat_1 \ahat_1 \rangle_\psi + 1) + \mathcal{O}(\theta^3)$ and it is thus proportional to the average photon number of the state on which we wish to perform the photon addition. Therefore, up to an error proportional to $\theta^2$, we can replace the photon-addition elements $\adhat$ in \autoref{fig:main-circuit} by adding an equal number of modes prepared in the Fock state $\ketbra{1}{1}$, combine all the inputs into a single, larger interferometer with $N + M + 2$ input modes, described by the unitary $\hat{\bm{U}}$, and then conditioning on $N$ output modes to be in the vacuum state. This rewriting of our procedure, which identifies it as an instance of boson sampling with post-selection and is sketched schematically in \autoref{fig:main-circuit-boson}, makes it clear that we are effectively injecting $N$ photons in the output and recovering $N - d$ in the PNR projection on one ancillary mode, so that exactly $d$ photons are staying in the final state (before dehomogenization). The extra photons that have to be injected and then measured act as \emph{catalyzers} for the preparation of the most generic multiphoton state, an idea upon which we shall elaborate further after \autoref{thm:photon-catalysis}.
	
	A different scheme to implement  photon addition, perhaps more commonly employed in all-optical platforms, where single-photon sources might not be readily available, relies on a two-mode squeezing operation, performed by a nonlinear crystal, and post-selection conditioning on having one photon on the ancillary mode \cite{Roeland_2022}. In this case, the ancillary input mode is in the vacuum state and no additional photons have to be injected: the energy is provided by the squeezing operation, through the pump. This option allows us to rewrite the circuit as a post-selected Gaussian boson sampling as depicted in \autoref{fig:main-circuit-gaussian-boson}. Notice, in particular, that now we have the vacuum state of $N + M + 2$ modes followed by a general multimode \emph{Gaussian unitary operation} $\hat{\bm{G}}$. Using the Bloch-Messiah decomposition \cite{PhysRevA.71.055801} and taking into account that any interferometer whose input modes are in the vacuum leaves them in the vacuum, we can rewrite $\hat{G}$ acting on $\vert 0 \rangle^{\otimes (N + M + 2)}$ as $N$ weak single-mode squeezers acting on the vacuum states of the ancillary modes, followed by a general interferometer and finally by the conditioning projective measurements, on $\ketbra{1}{1}$ for the $N$ ancillary modes used to implement the approximate addition and on $\ketbra{n}{n}$ for the ancillary mode needed to prepare $\vert \psi^{h} \rangle$. In this case, no non-Gaussian resource is used before the measurements, but the number of catalyzing photons is doubled, since the process underlying single-mode squeezing always produces photons in pairs.
	
	\section{Assessment of the efficiency of the method}\label{sec:bounds}
	\subsection{Bounds on resource use}
	
	Using the Waring decomposition introduced in \ref{th:main-theorem},
    the number of photon additions to be performed is given by $r \times d$, hence the order of the PNR projection upon which one has to condition to obtain the desired state is $d(r-1)$: in physical terms, $r \times d$ photons are injected into the interferometer and $d(r-1)$ of them are recovered when post-selecting on the appropriate Fock state of the ancillary mode, such that the final state contains exactly $d$ photons. Both these quantities scale linearly with the Waring rank $r$ of the associated polynomial. Thus, to assess the cost of generating a state with our method we need to understand how $r$ depends on the number of modes (resp., of variables) $M$ and on the number of photons (resp., degree of the polynomial) $d$.
	
	While the value of $r$ depends on the particular polynomial, the concept of \emph{generic rank} provides a way of reasoning about a general or \enquote{typical} polynomial. The value of the generic rank, as a function of the degree $d$ and number of variables $M$, $\bar{r}_{\mathrm{gen}}\left(M, d\right)$ is given by the Alexander-Hirschowitz theorem \cite{Landsberg2012_AH}: 
	
	\begin{equation}
		\bar{r}_{\mathrm{gen}}\left(M, d\right) = \left\lceil \frac{1}{M} \binom{M + d - 1}{d} \right\rceil
	\end{equation}
	
	The precise statement of the theorem and a geometrical interpretation of the generic rank are presented in the Appendix (see \autoref{th:alexander-hirshowitz}). From a physical point of view, the generic rank can be interpreted as follows:
	
	
	\begin{itemize}
		\item A randomly selected homogeneous core state has generic rank with probability $1$, but there exist states with ranks that are both higher (supergeneric) or lower (subgeneric) (this is implied by the fact that the set of tensors of non-generic rank has measure zero, see \autoref{th:measure-zero-lemma}).
		\item Any homogeneous core state, even of a supergeneric rank, could be approximated by a homogeneous core state of generic rank with fidelity arbitrarily close to 1.
	\end{itemize}
	
	Asymptotically, for $M \gg d$, we have the following scaling for the number of photon additions: 
	\begin{equation}
		d \times \bar{r}_{\mathrm{gen}}(M,d) \ = \ \dfrac{M^{d-1}}{(d-1)!}\left( 1 + O\left( d^2 / M \right) \right)
	\end{equation}
	which is polynomial in $M$. Given the symmetry of the binomial coefficient, for $d \gg M$ one gets a similar behaviour with $d$ and $M$ exchanged, i.e. a polynomial in $d$. Finally, for $M= c\times d \gg 1$, we get:  
	\begin{equation}
		d \times \bar{r}_{\mathrm{gen}}(M,d) \ = \ \dfrac{a_{c}}{ \sqrt{2 \pi d }} (k_{c})^{d} \left( 1 + O\left( d^{-1} \right) \right)
	\end{equation}
	where $a_{c} =\dfrac{1}{c}\sqrt{\frac{c+1}{c}}$ and $k_{c} = \frac{(c+1)^{c+1}}{c^c} > 1$,
	which is instead an exponential scaling in $d$, but only in the case where both the number of modes and the number of photons are simultaneously large. 
	
	It is worth noting that practically computing the best low rank approximation, or even the exact rank of a given tensor is an NP-hard problem (see \autoref{sec:numerical-methods}). However, the maximal rank of a tensor is always upper bounded by twice its generic rank \cite{Blekherman2015-sv} and tighter upper bounds are known for particular values of $d$ and $M$ \cite{Bernardi_2018}. \autoref{tbl:tensor-rank} summarizes these results.
	
	\begin{table}[ht]
		\centering
		\renewcommand{\arraystretch}{1.5}
		\begin{tabular}{c|c|c|c}
			$d$ & $M$ & Maximal Rank & Photon Additions \\
			\hline
			3 & 3 & 5 & 15 \\
			4 & 3 & 7 & 28 \\
			5 & 3 & 10 & 50\\
			3 & 4 & 7 & 21 \\
			any & 3 & $\leq \left\lfloor \frac{d^2 + 6d + 1}{4} \right\rfloor$ & $\leq \left\lfloor \frac{d^3 + 6d^2 + d}{4} \right\rfloor$ \\
			any & any & $\leq 2 \bar{r}_{\mathrm{gen}}\left(M, d\right)$ & $\leq 2 d\bar{r}_{\mathrm{gen}}\left(M, d\right)$
		\end{tabular}
		\caption{Maximal rank and corresponding number of photon additions depending on the number of photons $d$ and number of modes $M$.}
		\label{tbl:tensor-rank}
	\end{table}
	
%
%
	We also have a straightforward corollary of \autoref{th:elementary-symmetric-decomposition} which lower bounds the number of photon additions:
	\begin{corollary}
		\label{thm:photon-catalysis}
		If a homogeneous multiphoton state with $d$  photons is described by $M\geq 2$ intrinsic modes, then the minimum number $N$ of photon additions required to prepare it, within the established constraints, is lower bounded by $\max (d,M)$. 
	\end{corollary}
	Here, the number of intrinsic modes is defined as the minimum number of modes, over all possible mode bases choices, that are not in the vacuum state (see \cite{RevModPhys.92.035005}). Notice that states in the class $\left[1^d\right]_M$ have at most $\min\left(d, M\right)$ intrinsic modes, since one can attach at most one independent mode to each of the $d$ photon additions. More generally, if $M >d$ and the state has $M$ intrinsic modes, \autoref{thm:photon-catalysis} implies that we have to create at least $N \geq M$ photons, $N-d$ of which will have to be retrieved by a post-selection on a PNR measurement outcome; it is thus appropriate to think about our protocol as a \emph{multiphoton catalysis} mediated by these $N-d$ additional photons which are used only to generate the required number of intrinsic modes, thus to prepare states outside the class $\left[1^d\right]_M$.

	\subsection{Probability of success}
	
	\begin{figure}[t]
		\begin{tikzpicture}
			\definecolor{mplblue}{rgb}{0.121, 0.466, 0.705}  
			\definecolor{mplorange}{rgb}{1.000, 0.498, 0.054} 
			\definecolor{mplgreen}{rgb}{0.172, 0.627, 0.172}  
			
			\begin{axis}[
				name=axis,
				xmode=linear,
				ymode=linear,
				ytick={0, 0.05, ..., 0.5},
				xtick={0, 0.5, ..., 5.0},
				xlabel={Scaling factor},
				ylabel={Probability of projection success}]
				\addplot [mplblue, domain=0:5,samples=201, style=thick] {6! * 8 * x^4 / (8! + 6!*8*x^4 + 4!*4*6*x^8 + 2!*8*4*x^12 + 16*x^16)};\label{pgf:e1p4} 
                
				\addplot [mplblue, dashed, style=thick] table[col sep=comma] {e2_psi_4.csv};\label{pgf:e2p4}
				
				\addplot [mplorange, style=thick] table[col sep=comma] {e1_psi_8.csv};\label{pgf:e1p8}
				\addplot [mplorange, dashed, style=thick] table[col sep=comma] {e2_psi_8.csv};\label{pgf:e2p8}
				
				\addplot [mplgreen, style=thick] table[col sep=comma] {e1_psi_9.csv};\label{pgf:e1p9}
				\addplot [mplgreen, dashed, style=thick] table[col sep=comma] {e2_psi_9.csv};\label{pgf:e2p9}
			\end{axis}
			
			\node[draw, fill=white, inner sep=0, below left=0.5em] at (axis.north east) {\small
				\begin{tabular}{ccl}
					Waring & ESP & \\
					\ref{pgf:e1p4} & \ref{pgf:e2p4} & $\ket{\Psi_4}$\\
					\ref{pgf:e1p8} & \ref{pgf:e2p8} & $\ket{\Psi_8}$\\
					\ref{pgf:e1p9} & \ref{pgf:e2p9} & $\ket{\Psi_9}$
			\end{tabular}};
		\end{tikzpicture}
		\caption{Probability of successfully detecting the right number of photons to obtain $\ket{\Psi_4}, \ket{\Psi_8}, \ket{\Psi_9}$ from \autoref{tbl:selection-of-states}. The solid lines show the final PNR detection probability for the method described in \autoref{th:main-theorem}, while the dashed lines show the PNR detection probability of the method described in \autoref{th:elementary-symmetric-decomposition}, both maximized over 25 possible decompositions. The scaling $1$ corresponds to the Frobenius norm $\|W\| = \sqrt{N}$.}\label{fig:probability}
	\end{figure}

	It is worth observing that the coefficients of $\lambda^{d(r-1)}$ in $L_{r,d,W}$ and in $L_{r,d, \alpha W}$ are identical, up to a scaling factor, and thus that the scaling coefficient $\alpha$ is a freely adjustable parameter, which can impact the probability of success. It can effectively steer the distribution of photon number measurements on the ancillary mode to peak at $d(r-1)$. Physically, this scaling factor corresponds to the relative weight between the ancilla and the rest of the modes in the state at the output of the addition part $L_{r,d, \alpha W} = \prod \left(\lambda^d + \alpha^d \left( \mbf{w}_k^T \mbf{x}\right)^d\right)$, thus changing the interferometers $\hat U_k$ in \autoref{fig:main-circuit}. This highlights that the choice of the interferometers in our state preparation protocol is not unique and can, therefore, be optimized. Practically calculating the probability of successfully projecting onto the required amount of photons in the ancillary mode can be done symbolically or numerically. For the simple but nontrivial case where $W \propto I$, the explicit formula can be provided to highlight the nonlinear behavior of such probability.

	\begin{corollary}
		If $W = wI_{M \times M}$ -- rescaled identity matrix, $w \in \mathbb{C}$, then the probability of successfully obtaining $P_\psi$ from $L_\psi$ by \autoref{th:main-theorem} is given by \autoref{eq:identity-prob} and plotted in \autoref{fig:probability}.
		\begin{equation}
			\label{eq:identity-prob}
			\operatorname{Prob} = \frac{\left(d\left(M - 1\right)\right)!d! M w^{2d}}{\sum_{k=0}^M \left(d\left(M - k\right)\right)! \left(d!\right)^k \binom{M}{k}  w^{2dk}}
		\end{equation}
	\end{corollary}
	
	\begin{proof}
		Denote $\stnorm{P_\psi}$ -- norm of the unnormalized quantum state corresponding to $P_\psi$ (this norm is equal, up to degree normalization, to the Bombieri \cite{BEAUZAMY1990219} norm).
		
		\begin{align}
			\stnorm{L_\psi}^2 &= \stnorm{\sum_{k=0}^M \lambda^{d\left(M - k\right)} e_k\left(w x_1, \dots, w x_M\right)}^2 \\
			&= \sum_{k=0}^M \left(d\left(M - k\right)\right)! \left(d!\right)^k \binom{M}{k} w^{2dk}
		\end{align}
		\begin{align}
			\stnorm{\lambda^{d\left(M - 1\right)}e_1\left(w\mbf{x}\right)}^2 = \left(d\left(M - 1\right)\right)!d! M w^{2d}
		\end{align}
		\begin{equation}
			\operatorname{Prob} = \frac{\stnorm{\lambda^{d\left(M - 1\right)}e_1\left(w\mbf{x}\right)}^2}{\stnorm{L_\psi}^2}
		\end{equation}
	\end{proof}
	
	Such nonlinear behavior of the probability implies that in the most general case, after finding the Waring decomposition of the target polynomial, an additional optimization should be carried out to find the scaling factor $\alpha$ maximizing the probability of success.

	\section{Special cases and examples}
	\label{sec:special-cases}
	
	As will be shown in this section, the method given in \autoref{th:main-theorem} is not always optimal in terms of resources, despite the fact that it is fully general. Some classes of states could be produced within the same framework given by \autoref{th:monic-and-waring} but with a smaller number of catalysis photons, using linear combinations of elementary symmetric polynomials, or elementary symmetric polynomials of higher degrees in place of $e_1$. We analyze in the following subsections some particular cases in which this is possible, and show how the bottleneck of the Waring decomposition can be avoided.
	
	\subsection{Linear combinations of $\{e_k\}$}
	
	Here we present a possible specialization of \autoref{th:main-theorem} for the states that are described by linear combinations of elementary symmetric polynomials. We show that those states could be obtained from \autoref{eq:L-sum} by performing a projection onto a superposition of terms on the ancillary mode, which is equivalent to a sequence of displacements and photon subtractions followed by vacuum projection.
	
	\begin{example}
		\label{ex:GHZ}
		Consider the GHZ state and its homogenization.
		\begin{align}
			\ket{\text{GHZ}} &= \frac{\ket{0}^{\otimes M} + \ket{1}^{\otimes M}}{\sqrt{2}} \\ \ket{\text{GHZ}^h} &= \frac{\ket{M}\ket{0}^{\otimes M} + \ket{0}\ket{1}^{\otimes M}}{\sqrt{2}} \\
			P_{\text{GHZ}^h} &= \frac{1}{\sqrt{2}}\left(\frac{1}{\sqrt{M!}} x_0^M + x_1 \dots x_M\right)
		\end{align}
		Note that the homogenization corresponds exactly to the first and last terms in the expansion of \autoref{eq:L-sum} for $W = I$, $r = M$, $d = 1$. Thus, by projecting this element of $\left[1^M\right]_M$ onto $\ket{\pi} \propto \left( \frac{1}{\sqrt{M!}} \ket{M} + \ket{0} \right)$, we will obtain $\ket{\text{GHZ}^h}$. To see this explicitly, we first express:
		\begin{multline}
			L_{M, 1, I} = \prod_{k = 1}^M\left(\lambda - x_k\right) = \sum_{k=0}^M \lambda^{M - k} e_k\left(\mbf{x}\right) \\
			= \lambda^M + x_1 \dots x_M + \sum_{k = 1}^{M - 1}  \lambda^{M - k} e_k\left(\mbf{x}\right).
		\end{multline}
		Now, we can easily evaluate the projection on $\ket{\pi}$:
		\begin{multline}
			\dbraket{\pi}{L_{M, 1, I}} \propto \left(\frac{1}{\sqrt{M!}}\bra{M} + \bra{0}\right)\\
			\left(\sqrt{M!}\ket{M}\ket{\mbf{0}} + \ket{0}\ket{\mbf{1}} + \sum_{k=1}^{M-1} \sqrt{k!} \ket{k}\left(\dots\right) \right) \\\propto \ket{\text{GHZ}^h}
		\end{multline}
		The projection $\bra{\pi}$ can be realized as follows. As any single variable polynomial with complex coefficients is fully factorizable, we could write $P_\pi$ as a product of factors:
		\begin{equation}
			P_\pi \propto x_0^M + M! = \prod_{j = 0}^{M - 1} \left(x_0 - r_j\right)
		\end{equation}
		where $r_{j}$ are the roots of $P_{\pi}$. Denoting by $\hat{D}\left(r \right) = \exp \{ r \hat{a}^{\dagger}_{0} - r^{*} \hat{a}_{0} \} $ the displacement operator, we rewrite:
		\begin{multline}
			\ket{\pi} \propto \prod_{j = 0}^{M - 1} \left(\adhat_0 - r_j\right) \ket{0} = \prod_{j=0}^{M-1}\hat{D}\left(r_j^*\right) \adhat_0 \hat{D}^{\dag}\left(r_j^*\right)\ket{0}
		\end{multline}
		Therefore, $\ket{\pi}$ could be prepared using $M$ displacements and $M$ photon additions. By formally reversing the circuit, we find the corresponding $\bra{\pi}$, which consists of displacements and photon subtractions followed by a vacuum projection:
		\begin{equation}
			\bra{\pi} \propto \bra{0}\hat{D}^{\dag}\left(r_{M-1}^*\right) \ahat_0 \hat{D}\left(r_{M-1}^*\right)\dots \hat{D}^{\dag}\left(r_{0}^*\right) \ahat_{0} \hat{D}\left(r_{0}^*\right)
		\end{equation}
		This example shows how using two terms of the monic expansion we are able to prepare the $M$-mode homogenized GHZ state using $M$ photon additions, $M$ photon subtractions, $M$ displacements and one vacuum projection. Directly using \autoref{th:main-theorem} would require us to use rank $r = 1 + 2^{M - 1}$ (which is the rank of two non-overlapping monomials \cite{RankOfMonomial}) and degree $d = M$, and thus inflict a much larger PNR count $M 2^{M - 1}$. Both methods should be followed by the dehomogenization procedure to obtain a true GHZ state.
	\end{example}

	\subsection{ESP decomposition for two-photon states}
	\label{sec:quadratic}
	
	Our protocol can be simplified significantly in case of two-photon states, i.e., core states of stellar rank two. This specific setting was already studied in literature \cite{10821417} using methods based on permanents. In \autoref{th:e2-corollary}, we show that Theorem II of \cite{10821417} follows very naturally from our algebraic framework, as a particular case of the decomposition described in our \autoref{th:elementary-symmetric-decomposition} for two-photon states. We thus find that any such state in $M$ intrinsic modes can be produced by injecting exactly $M$ photons. In the following corollary we explicitly give the coefficients of the corresponding decomposition by analyzing the linear transformations of quadratic forms.
	
	\begin{corollary}
		\label{th:e2-corollary}
		Suppose that the target state is described by $P_{\psi}\left(\mbf{x}\right) \in \mathbb{C}\left[x_1, \dots, x_M\right]$ of degree $2$. Then the order of the Fock state projection in the preparation of $\ket{\psi}$ can be reduced to $M^\prime- 2$, where $M^\prime$ is the rank of the quadratic form of the homogeneization of $P_{\psi}$.
		Furthermore, for homogeneous polynomials this is the shortest scheme involving the given resources (photon additions, multiport interferometers and PNR detection on one ancillary mode).
	\end{corollary}
	\begin{proof}
		The dehomogenization step could be applied similarly to \autoref{th:main-theorem}. Thus, suppose without loss of generality that $P_{\psi}$ is homogeneous with one extra mode, then it follows that it could be written as $P_{\psi} = \mbf{x}^T B \mbf{x}$ for a complex symmetric matrix $B^T = B$. Also, without loss of generality, we assume that $B$ is full rank, otherwise it is possible to apply a linear transformation that maps the variables to the essential variables -- a subspace of dimension $M^\prime < M$ such that $B^\prime$ is full rank. We first note that, for any symmetric (possibly complex) matrix $B$ there exists a unitary matrix $U$ and a diagonal matrix $D$ such that $B = U^T D U$ (given by Autonne-Takagi decomposition \cite{Horn_Johnson_2012}). By choosing a (generally non unique) complex square root of the diagonal matrix $D$, we can rewrite this factorization as $B = S_B^T S_B$ for some $S_B$. Let $S_A^T S_A$ be such a decomposition for the matrix $A$ that defines $e_2$:
		\begin{equation}
			\label{eq:e2-matrix}
			e_2\left(\mbf{x}\right) = \mbf{x}^T A \mbf{x}, \ \ \ \ A =
			\frac{1}{2}\begin{bmatrix}
				0 & 1 & 1 & \dots & 1 \\
				1 & 0 & 1 & \dots & 1 \\
				&   &   & \ddots & \\
				1 & 1 &   & \dots & 1 \\  
				1 & 1 & 1 & \dots & 0
			\end{bmatrix}
		\end{equation}
		
		From \autoref{eq:L-sum} in the proof of \autoref{th:monic-and-waring-0}, it follows for $r = M$ and $d = 1$:
		\begin{multline}
			\prod_{k=1}^{M}\left(\lambda + \mbf{w}_k^T \mbf{x}\right) = \sum_{k = 0}^M \lambda^{M - k} e_k\left(W \mbf{x}\right) =\\
			\lambda^{M - 2} e_2\left(W \mbf{x}\right) + \underbrace{\dots}_\text{other powers of $\lambda$}
		\end{multline}
		
		The coefficient of $\lambda^{M - 2}$ can then be rewritten as a quadratic form related by congruence (not necessarily unitary, as the matrix $W$ is arbitrary) to the initial quadratic form of $e_2\left(\mbf{x}\right)$:
		\begin{equation}
			e_2\left(W \mbf{x}\right) = \left(W \mbf{x}\right)^T A \left(W \mbf{x}\right) = \mbf{x}^T W^T A W \mbf{x}
		\end{equation}
		
		Therefore to prepare the desired state $\ket{\psi}$ it suffices to find $W$ such that
		\begin{align}
			&\mbf{x}^T W^T A W \mbf{x} = \mbf{x}^T B \mbf{x} \\
			&\iff W^T A W = \left(S_A W\right)^T \left(S_A W\right) = B = S_B^T S_B \\
			&\iff S_B = S_A W \\
			&\iff W = S_A^{-1} S_B
		\end{align}
		
		Note that $A$ is real by \autoref{eq:e2-matrix} and its symmetric part can be diagonalized as $A = T^T D T$ with $D = \operatorname{diag} \begin{bmatrix} 1 , \frac{-1}{M - 1} , \dots , \frac{-1}{M - 1}  \end{bmatrix}$, such that $S_A = \operatorname{diag}\begin{bmatrix} 1 , \frac{i}{\sqrt{M-1}} , \dots \frac{i}{\sqrt{M-1}} \end{bmatrix} T$.
		
		Since $S_A$ is invertible, such $W$ always exists and we conclude that any $2$-photon core state in $M$ modes can be prepared using one extra mode, $M$ photon additions and multiport interferometers, by conditioning on counting $M - 2$ photons in the ancillary mode. The dehomogenization procedure, when needed, just adds a second ancillary mode which is then displaced and post-selected on the vacuum.
		
		The optimality for an arbitrary state follows from the bound $N \geq M$ in \autoref{thm:photon-catalysis}.
	\end{proof}
	
	\autoref{th:e2-corollary} could be extended to $e_k, k \geq 3$, only by increasing the number of linear factors in $L$, which is demonstrated by counting the number of available parameters. Indeed, the number of parameters in $W$ that we can tune is fixed to $M^2$, but the total number of degree $d$ homogeneous polynomials in $M$ variables is
	\begin{equation}
		\label{eq:number-of-hom-polys}
		\begin{pmatrix}
			d + M - 1 \\
			d
		\end{pmatrix}
	\end{equation}
	Already for $d=3$ \autoref{eq:number-of-hom-polys} gives $\frac{M\left(M+1\right)\left(M+2\right)}{6} > M^2$ for $M > 2$. This highlights the importance of the decomposition introduced in \autoref{th:elementary-symmetric-decomposition}.
	
	\label{sec:examples}
	\begin{example}
		Consider the quantum state $\ket{\Psi_4}$ mentioned in \cite{Kopylov:2025zcs} and the corresponding stellar polynomial.
		\begin{equation}
			\ket{\Psi_4} = \frac{1}{\sqrt{4}}\left( \ket{2,0,0,0} + \ket{0,2,0,0} + \ket{0,0,2,0} + \ket{0,0,0,2} \right)
		\end{equation}
		\begin{equation}
			P_{\Psi_4}\left(x_1, x_2, x_3, x_4\right) = \frac{1}{2\sqrt{2}}\left(x_1^2 + x_2^2 + x_3^2 + x_4^2\right)
		\end{equation}
		
		We present two ways of preparing this state, first based on \autoref{th:main-theorem} and second based on \autoref{th:e2-corollary}. Note that the state is homogeneous so we omit the dehomogenization step and introduce only one ancillary mode.
		
		\paragraph{Using $e_1$.} Since the polynomial $P_{\Psi_4}$ is already presented as a sum of powers of its variables, we do not need to compute its Waring decomposition, we directly find $W = w I_{4\times4}$, where $w \approx 1.45$ according to \autoref{fig:probability}. Then the seed polynomial $L_{\Psi_4}$ is:
		\begin{align}
			L_{\Psi_4} &= \prod_{k=1}^4 \left(\lambda^2 + \left(w x_k\right)^2\right)\\
			&= \left(\lambda + i w x_1\right)\left(\lambda - i w x_1\right)\left(\lambda + i w x_2\right)\dots \\
			&= \sum_{k=0}^4 \lambda^{2\left(4 - k\right)} e_k\left(\left(w x_1\right)^2, \dots, \left(w x_4\right)^2\right) \\
			&= \lambda^8 + \lambda^6 w^2\left(x_1^2 + x_2^2 + x_3^2 + x_4^2\right) + \dots
		\end{align}
		
		The seed polynomial is factorizable and therefore the seed state can be prepared using photon additions and interferometers. Then, by conditioning on detecting $2\left(4 - 1\right) = 6$ photons on the anciliary mode $\lambda$, we obtain (up to normalization) the desired state $\ket{\Psi_4}$.
		
		\paragraph{Using $e_2$.} Denote $B$ the matrix corresponding to the quadratic form $P_{\Psi_4}\left(\mbf{x}\right) = \mbf{x}^T B \mbf{x}; B = I_{4 \times 4}$.
		Then we can find $W$ such that $W^T A W = B$:
		\begin{equation}
			W = w \left[\begin{matrix}i & \frac{\sqrt{3} i}{3} & \frac{\sqrt{6} i}{6} & \frac{\sqrt{6}}{6}\\- i & \frac{\sqrt{3} i}{3} & \frac{\sqrt{6} i}{6} & \frac{\sqrt{6}}{6}\\0 & - \frac{2 \sqrt{3} i}{3} & \frac{\sqrt{6} i}{6} & \frac{\sqrt{6}}{6}\\0 & 0 & - \frac{\sqrt{6} i}{2} & \frac{\sqrt{6}}{6}\end{matrix}\right]
		\end{equation}
		The normalization factor $w \approx 1.1$ maximizes the success probability according to \autoref{fig:probability}.
		Finally, we can prepare the state corresponding (after normalization) to the stellar polynomial $L_{\Psi_4}^\prime$ and condition on detecting $4 - 2 = 2$ photons on the anciliary mode.
		\begin{align}
			L_{\Psi_4}^\prime &= \prod_{k=1}^{4}\left(\lambda - \mbf{w}_k^T \mbf{x}\right)\\
			&= \sum_{k=0}^4 \lambda^{4 - k} e_k\left( W\mbf{x} \right)
		\end{align}
		
		The second method requires less catalyzing photons but achieves slightly lower probability of success.
	\end{example}

	\section{Numerical methods}
	\label{sec:numerical-methods}
	
	\begin{table}[ht]
		\centering
		\begin{alignat*}{2}
			&\ket{\Psi}_1 &&\propto \ket{200} + \ket{020} + \ket{002} \\
			&\ket{\Psi}_2 &&\propto \ket{300} + \ket{030} + \ket{003} \\
			&\ket{\Psi}_3 &&\propto \ket{400} + \ket{040} + \ket{004} \\
			&\ket{\Psi}_4 &&\propto \ket{2000} + \ket{0200} + \ket{0020} + \ket{0002} \\
			&\ket{\Psi}_5 &&\propto \ket{012} + \ket{120} + \ket{201} + \ket{021} + \ket{102} + \ket{210} \\
			&\ket{\Psi}_6 &&\propto \ket{110} + \ket{101} + \ket{011} \\
			&\ket{\Psi}_7 &&\propto \ket{220} + \ket{202} + \ket{022} \\
			&\ket{\Psi}_8 &&\propto \ket{2000} + \ket{0110} + \ket{0002} \\
			&\ket{\Psi}_9 &&\propto \ket{3000} + \ket{0210} + \ket{0120} + \ket{0003} \\
			&\ket{\Psi}_{10} &&\propto \ket{040} + \ket{121} + \ket{202}\\
			&\ket{R_2} &&\propto \ket{300} + \sqrt{3}\ket{120} + \sqrt{6}\ket{111} + \sqrt{3}\ket{102}\\
			&\ket{R_4} &&\propto \ket{300} + \ket{030} + \ket{003} + \ket{111} \\
			&\ket{R_5} &&\propto \ket{210} + \ket{021} \\
			&\ket{K_3} &&\propto \ket{3000} + \ket{2100} + \ket{2010} + \ket{2001} \\
			& && - \ket{1110} - \ket{1101} - \ket{1011} - \ket{0111}
		\end{alignat*}
		\caption{States used for evaluation}
		\label{tbl:selection-of-states}
	\end{table}
	
	While the first general method presented in \autoref{th:main-theorem} allows the preparation of the desired state given its Waring decomposition, practically computing this decomposition can be challenging. In particular, it is known that both determining the rank of a symmetric tensor and computing its best low-rank approximation are NP-hard problems \cite{TensorsAreNPHard}. It is also worth noting that the decomposition is not unique, except for some special cases \cite{GaluppiMellaIdentifiability}. Additionally, the number of catalysis photons predicted by \autoref{th:main-theorem} is larger than by \autoref{th:elementary-symmetric-decomposition}, therefore the latter should be used in practice. Nevertheless, for the sake of completeness, we present the implementation of both methods.
	
	The numerous applications of tensor decomposition across different fields have given rise to various decomposition algorithms. Purely algebraic methods, based on properties of the catalecticant matrix \cite{Bernardi_2018, Oeding_2013} can provide exact solutions when the decomposition rank is low, and are thus not generally applicable. A variety of numerical algorithms have also been proposed \cite{kileel2025subspacepowermethodsymmetric, lowrankapprox}, covering different performance trade-off, use cases or implementation platforms. Apart from the case $d=2$ presented in \autoref{sec:quadratic}, algebraic solutions to the ESP decomposition problem can also be obtained by expanding the expression $P(\mbf{x}) = e_d(U\mbf{x})$, yielding a system of polynomial equations in the coefficients of $U$, which can be solved using Gröbner bases methods (for example using SymPy's \texttt{solve\_poly\_system} or the \texttt{msolve} \cite{msolve} software package). The growing size of the bases makes this exact method applicable only for small values of $M$ and $N$.
	
	In this study, the decompositions were thus performed numerically, using the Optax optimization library \cite{deepmind2020jax} and its implementation of the Adam gradient descent algorithm. Starting from a target state, the optimization process is as follows:

    \paragraph{Waring decomposition.} First, the state is converted into the corresponding polynomial $P_{\Psi}$ in creation operators. This polynomial is represented as an order-$d$ symmetric tensor $T$ with $M^d$ complex entries. Assuming the rank $r$ is given, one can find by gradient descent the $r \times M$ matrix $W$ achieving the best approximation of $T$, as measured by $||T - \sum_{k=1}^r \mbf{w}_k^{\otimes d}||_2$. Since the rank $r$ is not known, we proceed iteratively, increasing the value of $r$ until a near-perfect reconstruction of $T$ is obtained. Then, once the rank and decompositions are found, the seed state $L_{r, d, \alpha W}(\mbf{\hat{a}}^\dagger) \ket{\mbf{0}}$ is considered, normalized, and its projection on $\ketbra{d(r-1)}{d(r-1)}_0$ determines the probability of success of the preparation scheme, i.e. the probability of having measured $d(r-1)$ photons on the anciliary mode. This probability depends on both $\alpha$ and the decomposition $W$. Given a decomposition $W$, the value of $\alpha$ maximizing the success probability is found by numerical optimization.

    \paragraph{ESP decomposition.} If the number of catalysis photons $c = N - d$ is fixed, the optimization process can directly search for the $N \times M$ matrix $U$ such that the state whose stellar polynomial is $e_d(U \mbf{x})$ has maximum fidelity with the target state $\ket{\Psi}$ after projection on $\ketbra{c}{c}_0$. All these operations~--~expansion of $e_d$, scaling of its coefficients, and projection~--~are differentiable and thus allow direct optimization by Optax. Because the number of catalysis photons needed to obtain perfect fidelity is not known, $N$ is iteratively increased until a fidelity of 1 is reached. In a second step, the probability of success is maximized by optimizing the scaling factor of $U$. Note that any function of the fidelity and probability of success, expressing a trade-off between these two quantities, could be used as well in the optimization process. We emphasize that in case the available resources are capped, the optimization can be performed with a value of $N$ that corresponds to the target number of photon additions. In particular, by choosing $N = d + 1$, the first stage of the optimization process (fidelity maximization) exactly recovers the results of \cite{Kopylov:2025zcs}, as a special instance.
    
    In order to estimate how the probability of success can vary with the decomposition $W$ or $U$, the optimization process is performed $K=25$ times with randomly seeded initial conditions.
	
	For the purpose of comparison, we also implemented the scheme described in appendix C of \cite{Kopylov:2025zcs}, which requires $M + 1$ modes, $d + 1$ photon additions, and conditioning on detecting exactly one photon on the anciliary mode. We verified that the fidelity metric found by our re-implementation matches the values reported by the original authors, and we additionally computed its probability of success. 
	
	The states used for evaluation are presented in \autoref{tbl:selection-of-states}. The states $\ket{\Psi}_1$ to $\ket{\Psi}_7$ are the examples studied in \cite{Kopylov:2025zcs}. $\ket{\Psi}_8$, $\ket{\Psi}_9$ and $\ket{\Psi}_{10}$ involve a higher number of modes or photons and pose an additional challenge to preparation. States $\ket{R_2}$, $\ket{R_4}$ and $\ket{R_5}$ are such that the corresponding tensor is respectively of subgeneric ($r = 2$), generic ($r = 4$), and maximal ($r = 5$) rank for the $M = 3, d = 3$ case, and illustrate how resource use scales with Waring rank. Finally, $\ket{K_3}$ corresponds to a fully connected graph state on 3 vertices. 

    \begin{table*}
		\centering
        \renewcommand{\aboverulesep}{0ex}
        \renewcommand{\belowrulesep}{0ex}
		\renewcommand{\arraystretch}{1.2}
		
		\begin{tabular}{l c c | c c c c c c | c c c c c c | c c c c c c}
			\toprule
			\multicolumn{3}{c|}{} 
			& \multicolumn{6}{c|}{Protocol of \cite{Kopylov:2025zcs}} 
			& \multicolumn{6}{c|}{Waring decomposition} 
			& \multicolumn{6}{c}{ESP decomposition} \\
			\midrule
			State            & d & M &  \#Add & PNR & $p_{min}$ & $p_{med}$ & $p_{max}$ & F 
			& \#Add & PNR & $p_{min}$ & $p_{med}$ & $p_{max}$ & F 
			& \#Add & PNR & $p_{min}$ & $p_{med}$ & $p_{max}$ & F \\
			\midrule
			$\ket\Psi_1$ & 2 & 3 & 3 & 1 & 0.04 & 0.22 & 0.28 & 1.00 & 6 & 4 & 0.01 & 0.08 & 0.23 & 1.00 & 3 & 1 & 0.11 & 0.25 & 0.29 & 1.00 \\
			$\ket\Psi_2$ & 3 & 3 & 4 & 1 & 0.05 & 0.14 & 0.26 & 1.00 & 9 & 6 & 0.17 & 0.17 & 0.17 & 1.00 & 4 & 1 & 0.09 & 0.18 & 0.27 & 1.00 \\
			$\ket\Psi_3$ & 4 & 3 & 5 & 1 & 0.03 & 0.15 & 0.38 & 0.95 & 12 & 8 & 0.09 & 0.09 & 0.09 & 1.00 & 6 & 2 & 0.04 & 0.11 & 0.19 & 1.00 \\
			$\ket\Psi_4$ & 2 & 4 & 3 & 1 & 0.08 & 0.22 & 0.28 & 0.75 & 8 & 6 & 0.03 & 0.11 & 0.28 & 1.00 & 4 & 2 & 0.13 & 0.17 & 0.21 & 1.00 \\
			$\ket\Psi_5$ & 3 & 3 & 4 & 1 & 0.03 & 0.15 & 0.29 & 1.00 & 9 & 6 & 0.20 & 0.20 & 0.20 & 1.00 & 4 & 1 & 0.05 & 0.17 & 0.29 & 1.00 \\
			$\ket\Psi_6$ & 2 & 3 & 3 & 1 & 0.08 & 0.20 & 0.31 & 1.00 & 6 & 4 & 0.02 & 0.11 & 0.25 & 1.00 & 3 & 1 & 0.16 & 0.25 & 0.32 & 1.00 \\
			$\ket\Psi_7$ & 4 & 3 & 5 & 1 & 0.01 & 0.12 & 0.29 & 1.00 & 24 & 20 & 0.00 & 0.01 & 0.02 & 1.00 & 5 & 1 & 0.01 & 0.15 & 0.30 & 1.00 \\
			$\ket\Psi_8$ & 2 & 4 & 3 & 1 & 0.08 & 0.22 & 0.29 & 0.83 & 8 & 6 & 0.02 & 0.12 & 0.26 & 1.00 & 4 & 2 & 0.12 & 0.18 & 0.22 & 1.00 \\
			$\ket\Psi_9$ & 3 & 4 & 4 & 1 & 0.07 & 0.17 & 0.28 & 0.88 & 12 & 9 & 0.15 & 0.15 & 0.15 & 1.00 & 5 & 2 & 0.02 & 0.18 & 0.18 & 1.00 \\
			$\ket\Psi_{10}$ & 4 & 3 & 5 & 1 & 0.00 & 0.11 & 0.32 & 0.98 & 24 & 20 & 0.00 & 0.01 & 0.02 & 1.00 & 6 & 2 & 0.00 & 0.06 & 0.14 & 1.00 \\
			$\ket{R_4}$ & 3 & 3 & 4 & 1 & 0.04 & 0.19 & 0.27 & 1.00 & 12 & 9 & 0.00 & 0.03 & 0.10 & 1.00 & 4 & 1 & 0.08 & 0.24 & 0.30 & 1.00 \\
			$\ket{R_5}$ & 3 & 3 & 4 & 1 & 0.01 & 0.13 & 0.31 & 1.00 & 15 & 12 & 0.01 & 0.04 & 0.08 & 1.00 & 4 & 1 & 0.01 & 0.16 & 0.31 & 1.00 \\
			$\ket{R_2}$ & 3 & 3 & 4 & 1 & 0.07 & 0.34 & 0.71 & 1.00 & 6 & 3 & 0.13 & 0.13 & 0.13 & 1.00 & 4 & 1 & 0.00 & 0.40 & 0.82 & 1.00 \\
			$\ket{K_3}$ & 3 & 4 & 4 & 1 & 0.02 & 0.25 & 0.29 & 0.95 & 15 & 12 & 0.13 & 0.13 & 0.13 & 1.00 & 5 & 2 & 0.00 & 0.01 & 0.18 & 1.00 \\
			\bottomrule& 
		\end{tabular}	
		\caption{Resources (number of photon additions, and photon number to be resolved), minimum, median and maximum probability of success, and fidelity of the scheme presented in \cite{Kopylov:2025zcs} and the two proposed state preparation methods, for a selection of states. All methods require one anciliary mode.}
		\label{tbl:results}
	\end{table*}
	
	The results are detailed in \autoref{tbl:results}. Since the state preparation methods we propose are exact, the fidelity is always 1, with minor deviations attributable to numerical noise.
    While guaranteeing the preparation of an arbitrary state with $100\%$ fidelity, the Waring decomposition seems to always require more catalyzing photons with respect to the optimal ESP decomposition. The scheme in \cite{Kopylov:2025zcs} can attain perfect fidelity for a low number of modes and photons, but its fidelity already degrades when a fourth mode is introduced. In all the cases, the optimization process can converge to multiple solutions corresponding to the same fidelity for each method, yielding different success probabilities. A practical implementation should thus select the best of several candidate solutions, or include the success probability as an additional term in the function being optimized.
	
	We also compared the scheme in \cite{Kopylov:2025zcs} with the method tailored for quadratic polynomials presented in \autoref{sec:quadratic}, for an increasing number of modes $M$. For each value of $M$, we considered 100 two-photon states with coefficients sampled randomly from a normal distribution. The results are given in \autoref{tbl:random-results} and confirm the optimality of the preparation scheme presented in \autoref{th:e2-corollary}. The method we propose requires a single photon projection for the case of three modes, but then generalizes to two photon states in any number of modes with linear scaling of the required target of PNR.
	
	The software implementation of the described methods is available on GitHub \footnote{\url{https://github.com/EQ15T/photon-catalysis}}. It includes an illustrative example that converts an ESP decomposition into the Boson sampling scheme of \autoref{fig:main-circuit-gaussian-boson}, which can be subsequently simulated by the discrete-variable photonic simulator Perceval \cite{heurtel2023perceval}. The decomposition leading to the highest probability of success among 5 candidates is chosen, converted into a sequence of unitaries with the algorithms described in \autoref{App:VectorToUnitary}, and each ideal photon addition block is replaced by a low-reflectivity beam-splitter injecting a photon from an ideal single-photon source. As shown in \autoref{sec:asymptotic-photon-addition}, the reflectivity parameter can be tuned to balance fidelity and success probability. This tradeoff is illustrated in \autoref{fig:boson-sampling-tradeoff}, with ten distinct beam-splitter ratios logarithmically spaced between 95:5 and 50:50. For the simpler states, a fidelity above 99\% can be obtained with $p \geq 10^{-3}$, but for $\ket{\Psi}_{10}$, which has 4 modes and requires two catalysis photons, the probability of success drops below $10^{-6}$ to reach this fidelity target. This additional catalysis step incurs a cost in probability of success, yet it is an unavoidable cost: restricting the number of catalysis photons to one bounds the fidelity at 97\% even with the most accurate (and thus unlikely to succeed) photon addition.

    \begin{figure}[t]
    \def\mathdefault#1{#1} 
    \resizebox{0.5\textwidth}{!}{\input{"figures/selected_states.pgf"}}
    \caption{Probability of success and accuracy (measured by $1 - F$, equivalently the square of the trace distance) of the boson sampling implementation, using a non-ideal photon addition with a varying beam-splitter reflectivity. For $\ket{\Psi}_{10}$, optimal circuits with two catalysis photons (red) and only one (purple) are both considered.}
    \label{fig:boson-sampling-tradeoff}
    \end{figure}
    
	\begin{table*}
		\centering
		\aboverulesep=0ex
		\belowrulesep=0ex
		\renewcommand{\arraystretch}{1.2}
		\begin{tabular}{c c | c c c c c | c c c c c | c c c}
			\toprule
			\multicolumn{2}{c|}{} & \multicolumn{5}{c|}{Baseline} & 
			\multicolumn{5}{c|}{ESP decomposition, $N=4$} & 
			\multicolumn{3}{c}{ESP decomposition} \\
			\midrule
			d & M & \#Add. & PNR & F min & F avg & F max & \#Add. & PNR & F min & F avg & F max & \#Add. & PNR & F \\
			\midrule
			2 & 3 & 3 & 1 &    1 &    1 &    1 & 4 & 2 &   1 &   1 &    1 & 3 & 1 & 1 \\
			2 & 4 & 3 & 1 & 0.88 & 0.97 &    1 & 4 & 2 &   1 &   1 &    1 & 4 & 2 & 1 \\
			2 & 5 & 3 & 1 & 0.84 & 0.93 &    1 & 4 & 2 & 0.94&0.99 &    1 & 5 & 3 & 1 \\
			2 & 6 & 3 & 1 & 0.77 & 0.87 & 0.96 & 4 & 2 & 0.89&0.96 &    1 & 6 & 4 & 1 \\
			\bottomrule
		\end{tabular}
		\caption{Resources (number of photon additions, and photon number to be resolved) and fidelity of the proposed state preparation method with a fixed number of photon additions ($N = 4$), or with the number of photon additions guaranteeing perfect fidelity ($N = M$), compared to that of \cite{Kopylov:2025zcs}, for 2-photon states in $M$ modes with coefficients sampled randomly from a normal distribution.}
		\label{tbl:random-results}
	\end{table*}

	\section{Conclusions and outlook}


    In this work, we have proposed two methods to prepare an arbitrary multimode multi-photon state using a fixed set of operations that can be realistically implemented in a quantum-optical experiment: multi-port interferometers, photon additions, photon subtractions and vacuum projections, or equivalently post-selecting on the outcome of a PNR detection. Both methods are formulated as polynomial decompositions, the first being the well-studied Waring form, whereas the second is a novel decomposition. Results from algebraic geometry provide bounds on the resources used by the former and prove the existence of the latter, which is optimal. For the specific case of states having two photons in $M$ modes, we provide an explicit solution to this decomposition problem, which can in the most general case be solved by numerical optimization techniques. The proposed methods lead to circuits that can be built from bulk optical elements, but also formulated as instances of (Gaussian) boson sampling problems, compatible with the architecture of leading integrated photonic computing platforms \cite{larsen_integrated_2025}, \cite{Quandela24}. 

    We conjecture that a generalization of \autoref{th:main-theorem}, allowing to put unitary transformations between photon subtractions as in \autoref{ex:GHZ}, could help achieve a higher probability of success with less resources. In addition, further work is needed to study the algebraic properties of the Elementary Symmetric Polynomial decomposition of \autoref{th:elementary-symmetric-decomposition}, first to find suitable decomposition algorithms that do not rely on numerical optimization, but more importantly to find bounds on its rank as a function of both $d$ and $M$. While the Alexander-Hirschowitz theorem shows a rapid explosion of the generic rank as a function of these two parameters for the Waring decomposition, it is unclear how the resources scale for the more parsimonious ESP decomposition. The ``photon by photon" method we propose might not be the most efficient to engineer states with a large number of modes and photons, for which it could be more efficient to sculpt larger resource states by non-Gaussian operations. However, considering that states within experimental reach have only a limited number of modes and photons, we are far away from the practical limitation of our method. The performance vs fidelity tradeoff highlighted in \autoref{sec:numerical-methods} could be further refined by taking into account the optical losses or imperfect gate fidelities of existing photonic quantum computing platforms. Their ability to herald PNR detection patterns for more than one single photon makes them increasingly well-suited to our scheme. Performances might also be evaluated on mixed states under the presence of different types of noise. In parallel to these experimental concerns, the relation between the rank of these decomposition and the physical properties of entanglement, particularly whether it is Gaussian or not, is left to future research.

	Finally, we emphasize that, in a continuous-variable setting, any finite-stellar-rank state can be obtained by performing a well-chosen Gaussian operation on a core state, which is exactly a multimode multiphoton state. This implies that our methods are straightforwardly extended to the generation of any state of finite stellar rank by adding a Gaussian unitary. While, a priori, this might be impractical from an experimental point of view, adding such a Gaussian unitary transformation to the implementations of \autoref{sec:photonic-impl} will give rise to a more feasible experimental design for continuous-variable quantum information processing \cite{PhysRevResearch.3.033018,PhysRevLett.130.090602}.
	
	\begin{acknowledgements}   
		We acknowledge fruitful discussions with Valentina Parigi, Carlos E. Lopetegui-Gonz\'alez, and Polina Kovalenko. A.A. and E.G. are   supported   by   the   European   Research Council under the Consolidator Grant COQCOoN (Grant No.  820079). M.W. and M.F. acknowledge financial support from the ANR JCJC project NoRdiC (ANR-21-CE47-0005). M.W. received financial support from Plan France 2030 through the project OQuLus (ANR-22-PETQ-0013).
	\end{acknowledgements}

	\newpage
	\appendix
	
	\section{Alexander-Hirschowitz theorem}
	\begin{theorem}[Alexander-Hirschowitz \cite{Landsberg2012_AH}]
		\label{th:alexander-hirshowitz}
		\
		
		Let $\mathcal{Z}_r\left(M, d\right) = \left\{\mathcal{P} \in S^d\left(\mathbb{C}^M\right) \mid \rank_S\mathcal{P} = r \right\}$. $r$ is called generic symmetric rank if it is the minimal rank such that $\mathcal{Z}_r\left(M, d\right)$ is dense in $S^d\left(\mathbb{C}^M\right)$ with respect to Zariski topology on $S^d\left(\mathbb{C}^M\right)$ viewed as a vector space over $\mathbb{C}$, i.e.
		\begin{equation}
			\bar{\mathcal{Z}}_r\left(M, d\right) = S^d\left(\mathbb{C}^M\right)
		\end{equation}
		
		For $d > 2$, the generic symmetric rank of $S^d\left(\mathbb{C}^M\right)$ is
		\begin{equation}
			\bar{r}_{\mathrm{gen}}\left(M, d\right) = \left\lceil \frac{1}{M} \binom{M + d - 1}{d} \right\rceil
		\end{equation}
		except for the cases $\left(M, d\right) \in \left\{ \left(3, 5\right), \left(4, 3\right), \left(4, 4\right), \left(4, 5\right) \right\}$, where it should be increased by 1.
	\end{theorem}
	
	The geometric and physical interpretations of the generic symmetric rank follow from the following statement about Zariski topology (see \cite[Proposition~4.9.5.1]{Landsberg2012_zariski})
	\begin{lemma}
    \label{th:measure-zero-lemma}
		Any Zariski closed proper subset of a projective space $\mathbb{P}V$ has measure zero with respect to any measure on $\mathbb{P}V$ compatible with its linear structure.
		
		In particular, a tensor not of typical rank has probability zero of being selected at random.
	\end{lemma}

	Finally note that a proper set $\mathcal{S} \subset V$ of measure zero, according to the measure induced by the Lebsegue one on projective spaces, has no interior points so each point $p \in \mathcal{S}$ from such set lies on its border. It follows that any neighborhood $U \ni p$ contains points outside of the set: $U \cap \left( V \setminus \mathcal{S} \right) \neq \varnothing$, thus it is possible to construct a sequence of points $q_i \in V \setminus \mathcal{S}$ converging to $p$.

	
	\section{Proofs}
	\begin{proof}[Proof of \autoref{th:monic-and-waring}]
		The equality follows from the fact that any homogeneous polynomial in two variables is factorizable into a product of linear terms, and particularly
		\begin{equation}
			a^d + b^d  \ = \ \prod_{j = 0}^{d - 1}\left( a - \omega_{2d} \omega_d^j b \right)
		\end{equation}
		since $a = \omega_{2d} \omega_{d}^{j} b$ for $j=0,...,d-1$ are the $d$ roots of $a^d +b^d = 0$, with $\omega_{d} = e^{2 \pi i / d  }$ is the fundamental $d$-root of unity. 
		It follows that every factor in $L_{r, d, W}$ can be factorized in such way, thus $L_{r, d, W} \in \left[1^N\right]_M$.
	\end{proof}
	
	\begin{proof}[Proof of \autoref{th:dehomogeniztion}]
		Let $P_{\psi}\left(\mbf{x}\right) = \sum_{\mbf{n}} p_{\mbf{n}} \mbf{x}^{\mbf{n}}$ be the stellar polynomial of the desired non-homogeneous state. The homogenization is constructed by adding an ancillary variable to each monomial to complement its degree to the degree of the polynomial, i.e
		\begin{equation}
			P_{\psi^h}\left(\lambda, \mbf{x}\right) = \sum_{\mbf{n}} p_\mbf{n} \lambda^{d - \left|\mbf{n}\right|} \mbf{x}^\mbf{n}
		\end{equation}
		It is then clear that $P_{\psi^h}\left(1, \mbf{x}\right) = P_{\psi}\left(\mbf{x}\right)$.
		
		Now let's consider the action of displacement operator $ \hat{D}_0\left(\alpha\right)$ on the ancillary mode of $\ket{\psi^h}$.
		\begin{align}
			\hat{D}_0\left(\alpha\right) \ket{\psi^h} &= \hat{D}_0\left(\alpha\right) P_{\psi^h}\left(\adhat_0, \mbf{\adhat}\right) U^\dagger U\ket{0\mbf{0}} \\
			&= \sum_{\mbf{n}} p_\mbf{n} \hat{D}_0\left(\alpha\right) \left(\adhat_0\right)^{d - \left|\mbf{n}\right|} \left(\mbf{\adhat}\right)^\mbf{n} U^\dagger \ket{\alpha}\ket{\mbf{0}}\\
			&= \sum_{\mbf{n}} p_\mbf{n} \left(\adhat_0 + \alpha\right)^{d - \left|\mbf{n}\right|} \left(\mbf{\adhat}\right)^\mbf{n} \ket{\alpha}\ket{\mbf{0}}\\
			&= \sum_{\mbf{n}} p_\mbf{n} \alpha^{d - \left|\mbf{n}\right|} \left(\mbf{\adhat}\right)^\mbf{n} \ket{\alpha}\ket{\mbf{0}} + \adhat_0 \left(\dots\right)\ket{\alpha}\ket{\mbf{0}}
		\end{align}
		
		After the projection onto the vacuum, the terms that are multiplied by the creation operator on the ancillary mode vanish: $\bra{0}\adhat_0\ket{\phi} = 0$, thus after the substitution of $\alpha = 1$ we obtain $\ket{\psi}$
		\begin{align}
			\bra{0} \hat{D}_0\left(\alpha\right) \ket{\psi^h} &= \sum_{\mbf{n}} p_\mbf{n} \alpha^{d - \left|\mbf{n}\right|} \left(\mbf{\adhat}\right)^\mbf{n} \ket{\mbf{0}} \dbraket{0}{\alpha} \\
			&\propto \sum_{\mbf{n}} p_\mbf{n} \alpha^{d - \left|\mbf{n}\right|} \left(\mbf{\adhat}\right)^\mbf{n} \ket{\mbf{0}} \\
			&\stackrel{\alpha = 1}{=} \sum_{\mbf{n}} p_\mbf{n} \left(\mbf{\adhat}\right)^\mbf{n} \ket{\mbf{0}} = \ket{\psi}
		\end{align}
	\end{proof}
	
	\begin{proof}[Proof of \autoref{th:elementary-symmetric-decomposition}]
		The proof of the first statement proceeds by demonstrating that the Waring decomposition is a special case of the introduced elementary symmetric decomposition.
		Consider $L_{r,d,W}\left(\mbf{x}\right)$ and rewrite \autoref{eq:L-prod}
		\begin{align}
			L_{r, d, W} &= \sum_{k=0}^{r} \lambda^{d\left(r - k\right)}e_k\left( \left(\mbf{w}_1^T\mbf{x}\right)^d, \dots, \left(\mbf{w}_r^T\mbf{x}\right)^d \right) \label{eq:app-e1-d} \\
			&= \prod_{k=1}^r \prod_{i=0}^{d-1} \left( \lambda - \omega_{2d}\omega_d^i \mbf{w}_k^T\mbf{x} \right) \\
			&= \prod_{s=1}^{rd} \left(\lambda + \mbf{u}_s^T \mbf{x} \right) \\
			&= \sum_{s = 0}^{rd} \lambda^{rd - s} e_s\left(\mbf{u}_1^T \mbf{x}, \dots, \mbf{u}_{rd}^T \mbf{x}\right) \label{eq:app-es-1}
		\end{align}
		
		As the coefficients in front of different powers of $\lambda$ are unchanged, we obtain the following equality for $k=1$ in \autoref{eq:app-e1-d} and $s = d$ in \autoref{eq:app-es-1}
		\begin{multline}
			e_1\left(\left(\mbf{w}_1^T\mbf{x}\right)^d, \dots, \left(\mbf{w}_r^T\mbf{x}\right)^d\right) = \sum_{k = 1}^r\left(\mbf{w}_k^T \mbf{x}\right)^d \\
			= e_{d}\left(\mbf{u}_1^T \mbf{x}, \dots, \mbf{u}_{rd}^T \mbf{x}\right) = e_d\left(U \mbf{x}\right)
		\end{multline}
		
		Therefore, by choosing $W$ to define the Waring decomposition of the desired homogeneous polynomial $P$, we have constructed the desired $U$ with $N = rd$ rows, which proves the existence of the decomposition.
		
		The second statement is similar to the preparation with Waring decomposition. Suppose the homogenized stellar polynomial $P_{\psi^h}$ has the decomposition
		\begin{equation}
			P_{\psi^h}\left(\mbf{y}\right) = e_d\left(U\mbf{y}\right) = e_d\left(\mbf{u}_1^T \mbf{y}, \dots, \mbf{u}_N^T \mbf{y}\right)
		\end{equation}
		Then we prepare the intermediate state with ancillary mode denoted by $\lambda$ by alternating multiport interferometers and photon additions
		\begin{equation}
			L^\prime\left(\lambda, \mbf{y}\right) = \prod_{k=1}^N\left(\lambda  + \mbf{u}_k^T\mbf{y}\right) \in \left[1^N\right]_{M+1}
		\end{equation}
		Applying the monic expansion (see \autoref{eq:monic-expansion}) we write
		\begin{equation}
			L^\prime\left(\lambda, \mbf{y}\right) = \sum_{k=0}^N \lambda^{N-k}e_k\left(U\mbf{x}\right)
		\end{equation}
		Then, from the correspondence in \autoref{tbl:state-poly-corresp}, conditioning the ancillary mode $\ahat_\lambda$ on having $N-d$ photons corresponds to collapsing the sum onto the coefficient next to $\lambda^{N-d}$, which gives exactly the desired polynomial $P_{\psi^h}$.
		
		Finally, by \autoref{th:dehomogeniztion}, applying displacement and vacuum-projection on mode $\adhat_y$ to $P_{\psi^h}\left(\mbf{y}\right)$ produces the dehomogenized polynomial $P_\psi\left(\mbf{x}\right)$.
		
		The optimality follows from considering the most general form of the stellar polynomial corresponding to the state after a series of multiport interferometers and photon additions with an ancillary mode, that is the most general $L\left(\lambda, \mbf{x}\right) \in \left[1^N\right]_{M+1}$. We can always rewrite such $L$ grouping linear factors that contain the ancillary variable and that do not. Formally, denoting $\left\{\mbf{u}_k\right\}$ the set of $f_1$ vectors inducing the first group of linear forms and $\left\{\mbf{v}_k\right\}$ the set of $f_2$ vectors inducing the second group:
		\begin{multline}
			L = \prod_{k = 1}^N L_k = \prod_{\mathclap{\lambda \in \dom L_k}} L_k \;\;\;\; \prod_{\mathclap{\lambda \not\in \dom L_j}} L_j \\
			\propto \prod_{k = 1}^{f_1}\left( \lambda + \mbf{u}_k^T \mbf{x}\right) \prod_{k=1}^{f_2}\left(\mbf{v}_k^T\mbf{x}\right)
		\end{multline}
		Then, conditioning the ancillary mode on having $f_1-d$ photons gives a coefficient of the corresponding power of $\lambda$ in this expression, that is 
		\begin{equation}
			P\left(\mbf{x}\right) = e_d\left(U\mbf{x}\right) \prod_{k=1}^{f_2}\left(\mbf{v}_k^T\mbf{x}\right)
		\end{equation}
		Finally we note that for an irreducible target polynomial $f_2 = 0, f_1 = N$, as otherwise there are at least two factors in the expression of $P$. By taking the smallest $N$ such that the decomposition exists, we minimize the number of catalysis photons, concluding the optimality.
		
		
		
	\end{proof}
	
	\begin{proof}[Proof of \autoref{thm:photon-catalysis}]
		Suppose that $P\left(\mbf{x}\right)$ is a homogeneous polynomial having $M$ intrinsic modes. Any polynomial in $L \in \left[1^N\right]_{M+n}$ with $n\geq 0$ ancillary modes can have at most $N$ intrinsic modes, since we can at most attach one intrinsic mode to each linear form in $L$. 
		By conditioning on specific powers of the ancillary variables, we always obtain a linear combination of products of the linear forms appearing in $L$, so each polynomial in $M$ variables that we can obtain from the coefficients of the ancillary variables in $L$ cannot contain more than $N$ intrinsic modes. It thus follows that if the target polynomial has $M$ intrinsic modes, we must have $N \geq M$. Moreover, it is also clear that the degree $d$ of $P$ is a lower bound to the number of linear factors $N$ in $L$, if $L$ has to generate $P$. Thus $N \geq d$ and, altogether, we have $N \geq \max (d,M)$.  
		
	\end{proof}
	
	\section{Construction of the sequence of interferometers}\label{App:VectorToUnitary}
	Here we give the algorithm, first presented in \cite{Kopylov:2025zcs}, that could be used to obtain a sequence of unitaries $\left\{U_i\right\}$ given a sequence of linear forms in creation operators $\left\{\mbf{w}_i^T \mbf{\adhat}\right\}_{i = 1..N}$, such that alternating photon additions on the ancillary mode $\adhat_0$ with these unitaries produces the same state as $\prod_{i=1}^N \mbf{w}_i^T \mathbf{a}^\dagger$. In other words, this is the algorithm that allows to construct a circuit that prepares the seed state corresponding to the fully factorizable stellar polynomial $L_{r, d, W}$ (depicted as the blue part in \autoref{fig:main-circuit}). The \emph{UnitaryCompletion} sub-routine builds the unitary matrix that dispatches a newly added photon, to the mode superposition described by $w_i$. While any orthogonalization method can be used for this task, the proposed Algorithm 2 ensures that the unitary is as sparse as it needs to be, and acts trivially on the modes that are not in the support of $w_i$. We observed that decompositions obtained by algebraic methods often lead to sparse $\{w_i\}$, and sparse unitaries might lead to more efficient circuits. The further decomposition of each of these unitaries into a network of beam-splitters and phase-shifters can then be obtained, for instance, with Reck's or Clements' methods \cite{Reck, Clements}.
	
	
	
	\SetKw{KwFrom}{from}
	
	\begin{algorithm}[h]
		\caption{Product of linear forms to unitaries}\label{alg:backinversion}
		\For{$i$ \KwFrom 1 \KwTo $N$}{
			$U_{N - i - 1} \gets \text{UnitaryCompletion}(w_i)$\\
			\For{$j$ \KwFrom $i + 1$ \KwTo $N$}{
				$w_j \gets w_j U_{N - i + 1}^\dagger$
			}
		}
	\end{algorithm}	
	
	\begin{algorithm}[t]
		\caption{Unitary completion}\label{alg:completion}
		\KwIn{A vector $w \in \mathbb{C}^n$}
		\KwOut{A unitary matrix $U \in \mathbb{C}^{n \times n}$ such that $U^\dagger e_0 = w / ||w||$}
		$U \gets 0_{n \times n}$ \tcp{Initialize zero matrix}
		$S \gets \{i | w_i \neq 0\}$ \tcp{Support of $w$}
		$n_s \gets |S|$\\
		$V \gets \mathbb{I}_{n_s}$ \tcp{Identity matrix of size $n_s$}
		\For{$i$ \KwFrom $0$ \KwTo $n_s - 1$}{
			$v_{0, i} \gets w_{S[i]}$
		}
		$V \gets \text{Gram-Schmidt}(\text{V})$\\
		$C \gets S$ \tcp{Affected columns}
		$R \gets [0] \cup S[1:]$ \tcp{Affected rows}
		\For{$i$ \KwFrom $0$ \KwTo $n_s - 1$}{
			\For{$j$ \KwFrom $0$ \KwTo $n_s - 1$}{
				$u_{R[i], C[j]} \gets v_{i, j}$\\
			}
		}
		\For{$j$ \KwFrom $1$ \KwTo $n - 1$}{
			\If{$j \notin S$}{
				$u_{j, j} \gets 1$\\
			}
		}
		\If{$0 \notin C$}{
			$u_{S[0], 0} \gets 1$\\
		}
		\Return{$U$}
	\end{algorithm}

	\section{Asymptotic error of photon addition}
	\label{sec:asymptotic-photon-addition}
	\paragraph{Beam splitter.} Consider $\hat{U}_{BS} = e^{\theta\left( \adhat_1 \ahat_2 - \ahat_1 \adhat_2\right)}$.
	\begin{multline}
		\hat{U}_{\mathrm{BS}} = 1 + \theta \left( \adhat_1 \ahat_2 - \ahat_1 \adhat_2\right) \\+ \frac{\theta^2}{2}\left( \adhat_1 \ahat_2 - \ahat_1 \adhat_2\right)^2 + \mathcal{O}\left(\theta^3\right)
	\end{multline}
	\begin{align}
		\bra{0}_2 \hat{U}_{\mathrm{BS}} \ket{\psi}_1\ket{1}_2 = \sum_{n=0} \psi_n \bra{0}_2 \hat{U}_{\mathrm{BS}} \ket{n}_1\ket{1}_2
	\end{align}
	
	Notice that $\bra{0}_2 \left( \adhat_1 \ahat_2 - \ahat_1 \adhat_2\right)^2 \ket{n}_1\ket{1}_2 = 0$ so
	\begin{multline}
		\bra{0}_2 \hat{U}_{\mathrm{BS}} \ket{\psi}_1\ket{1}_2
		\\= \sum_{n=0} \psi_n \bra{0}_2 \left(\theta\left( \adhat_1 \ahat_2 - \ahat_1 \adhat_2\right) + \mathcal{O}\left(\theta^3\right)\right) \ket{n}_1\ket{1}_2
		\\= \theta \adhat_1 \ket{\psi} + \mathcal{O}\left(\theta^3\right) = \theta\left[\adhat_1 \ket{\psi} + \mathcal{O}\left(\theta^2\right)\right]
	\end{multline}
	
	\paragraph{Squeezer.} Consider $\hat{U}_{\mathrm{PDC}} = e^{\xi \left(\adhat_1 \adhat_2 - \ahat_1 \ahat_2\right)}$.
	\begin{multline}
		\hat{U}_{\mathrm{PDC}} = 1 + \xi \left(\adhat_1 \adhat_2 - \ahat_1 \ahat_2\right) \\+ \frac{\xi^2}{2}\left(\adhat_1 \adhat_2 - \ahat_1 \ahat_2\right)^2 + \mathcal{O}\left(\xi^3\right)
	\end{multline}
	\begin{equation}
		\bra{1}_2 \hat{U}_{\mathrm{PDC}} \ket{\psi}_1\ket{0}_2 = \sum_{n=0}\psi_n \bra{1}_2 \hat{U}_{\mathrm{PDC}} \ket{n}_1\ket{0}_2
	\end{equation}
	
	Since $\bra{1}_2 \left(\adhat_1 \adhat_2 - \ahat_1 \ahat_2\right)^2 \ket{n}_1\ket{0}_2 = 0$ as well
	\begin{multline}
		\bra{1}_2 \hat{U}_{\mathrm{PDC}} \ket{\psi}_1\ket{0}_2
		\\= \sum_{n=0} \psi_n \bra{1}_2 \left(\xi \left(\adhat_1 \adhat_2 - \ahat_1 \ahat_2\right) + \mathcal{O}\left(\xi^3\right)\right) \ket{n}_1\ket{0}_2
		\\\propto \xi \adhat_1 \ket{\psi} + \mathcal{O}\left(\xi^3\right) = \xi \left[\adhat_1 \ket{\psi} + \mathcal{O}\left(\xi^2\right)\right]
	\end{multline}

	
	\bibliography{apssamp}
	
\end{document}

%% file: base.tex
\usepackage{etoolbox}
\usepackage{braket}
\usepackage{dsfont}
\usepackage{amsmath}
\usepackage{xcolor}
\usepackage{float}
\usepackage{graphicx}
\usepackage{amsmath}
\usepackage{amssymb}
\usepackage{amsthm}
\usepackage{hyperref} 
\usepackage{xcolor}
\usepackage{braket}
\usepackage[caption=false]{subfig}
\usepackage{listings}
\usepackage{tikz}
\usepackage[utf8]{inputenc}
\usepackage{csquotes}
\usetikzlibrary{arrows.meta, positioning, cd}
\usepackage{booktabs}
\usepackage{natbib}
\usepackage{algorithm2e}
\usepackage{dcolumn}
\usepackage{bm}
\usepackage{faktor}
\usepackage{pgfplots}
\usepackage[normalem]{ulem}
\usepackage{booktabs}
\pgfplotsset{compat=1.8}
\RestyleAlgo{ruled} 

\provideboolean{no-quantikz}
\ifthenelse{\boolean{no-quantikz}}{}{\usepackage{quantikz}}

\tikzset{
	cd/.style={
		nodes={inner sep=4pt},
		row sep=1.8em, 
		column sep=2.4em
	}
}

\newcommand{\ahat}{\hat{a}}
\newcommand{\adhat}{\hat{a}^\dagger}

\newcommand{\mbf}[1]{\mathbf{#1}}

\newcommand{\ketbra}[2]{\ket{#1}\!\!\bra{#2}}
\newcommand{\dbraket}[2]{\braket{#1 \! \mid \! #2}}
\newcommand{\modsq}[1]{\left|{#1}\right|^2}
\newcommand{\stnorm}[1]{\left\Vert#1\right\Vert_{\ket{}}}

\DeclareMathOperator{\rank}{rank}
\DeclareMathOperator{\dom}{dom}

\newtheorem{theorem}{Theorem}
\newtheorem{lemma}{Lemma}

\newtheorem{corollary}{Corollary}

\theoremstyle{definition}
\newtheorem{example}{Example}